\documentclass[12pt]{article}
\usepackage[T1]{fontenc}
\usepackage{lmodern}

\usepackage{sectsty}
\allsectionsfont{\textsf}
\sectionfont{\centering\sffamily}
\subsectionfont{\centering\sffamily}
\subsubsectionfont{\centering\sffamily}

\usepackage{etoolbox}
\usepackage{ascmac}
\patchcmd{\abstract}{\sfshape\abstractname}{{\sf \abstractname}}{}{}

\usepackage{amssymb}
\usepackage{bm}
\usepackage[cmex10]{amsmath}
\usepackage{empheq}
\usepackage{amsthm}
\usepackage{mathtools}
\usepackage[longnamesfirst]{natbib}
\usepackage{xcolor}
\usepackage[%
 setpagesize=false,%
 bookmarks=true,%
 bookmarksdepth=tocdepth,%
 bookmarksnumbered=true,%
 colorlinks=true,
 citecolor=black,
 urlcolor=blue,
 linkcolor=blue,%
 pagebackref=true, %
 backref=page,
 pdftitle={The Uniqueness of Exponential Second-Order Expected Utility},%
 pdfsubject={Decision Theory},%
 pdfauthor={Yosuke Hashidate},%
 pdfkeywords={Relative Entropy, Kullback-Leibler divergence, Ambiguity Aversion, Rational Inattention, Model Misspecification}%
]{hyperref}
\usepackage{hypernat}
\usepackage[margin= 1in]{geometry}
\usepackage{latexsym}
\ifdefined\pdfoutput\ifnum\pdfoutput=1 \usepackage{graphicx}\else\usepackage[dvipdfmx, hiresbb]{graphicx}\fi\else\usepackage[dvipdfmx, hiresbb]{graphicx}\fi
\usepackage[utf8]{inputenc}
\usepackage{booktabs}
\usepackage{multirow}

\theoremstyle{plain}
\newtheorem{theorem}{Theorem}[section]
\theoremstyle{plain}
\newtheorem{proposition}{Proposition}[section]
\theoremstyle{plain}
\newtheorem{lemma}{Lemma}[section]
\theoremstyle{plain}

\theoremstyle{plain}

\theoremstyle{plain}
\newtheorem{corollary}{Corollary}[section]
\theoremstyle{definition}
\newtheorem{remark}{Remark}[section]
\theoremstyle{definition}

\theoremstyle{definition}

\theoremstyle{definition}

\theoremstyle{definition}

\theoremstyle{definition}
\newtheorem{example}{Example}[section]
\theoremstyle{plain}

\renewcommand*{\backref}[1]{}
\renewcommand*{\backrefalt}[4]{[%
    \ifcase #1 Not cited.%
          \or Cited on pp.~#2.%
          \else Cited on pp. #2.%
    \fi%
    ]}

\DeclareMathOperator*{\argmax}{arg\,max}
\DeclareMathOperator*{\argmin}{arg\,min}

\usepackage{tikz}
\usetikzlibrary{positioning,shapes.geometric,arrows.meta}

\usepackage{enumitem}

\begin{document}
\title{{\sf The Uniqueness of Exponential Second-Order Expected Utility}}
\author{\textsc{Yosuke Hashidate\thanks{Corresponding author. Affiliation: Faculty of Economics, Sophia University; Address: 7-1, Kioi-cho, Chiyoda-ku, Tokyo 102-8554, Japan; Email: {\tt hashidate@sophia.ac.jp}}}}
\date{Current Draft: \today}
\maketitle

\begin{abstract}
A decision-maker entertains several probability models of the state of the world, distrusts each of them, and is unsure which of them holds. She therefore carries two relative-entropy penalties: one on departures from a given model and one on departures from the reference she places on models. 
The literature that introduced this structure treats the equality of the two penalty parameters as a tractable special case. We show that it is in fact a necessary and sufficient condition: her criterion carries a single relative-entropy penalty if and only if the two parameters agree, that is, when she charges the same price for doubting a model as for doubting her weights on models.
Under that one condition the two layers collapse: the criterion is an exponential second-order expected utility based at the Bayesian predictive distribution, for every compact model set and every prior, with no large-sample limit anywhere. When the condition fails, the caution she displays blends her two parameters, with a weight equal to the share of an act's predictive variance that comes from disagreement among the models; since that share varies across acts, no single parameter and baseline represent her, and the same variation identifies the two parameters separately. The exponential aggregator and relative entropy are each forced: no other continuous and strictly increasing aggregator delivers the collapse, and within the power-divergence family only relative entropy does.
\end{abstract}

\section{Introduction}
\label{sec:intro}

A decision-maker who distrusts her own probability model --- a probability distribution over the states of the world --- has two entropic devices at her disposal, and both are now standard. Within a single model, relative entropy measures how far an alternative law can be from the posited one before it is dismissed as implausible, and penalizing that distance yields the multiplier preferences of \citet{HS_2001}, axiomatized by \citet{S_2011} and used throughout macroeconomics and finance. Across models, a smooth aggregator applied to model-contingent values yields the smooth ambiguity criterion of \citet{KMM_2005}, whose predictive formulation is due to \citet{DP_2022}. Each device has a clean behavioral foundation, each is tractable, and the two are usually deployed one at a time.

But a criterion that uses both at once carries two caution parameters rather than one, and the literature that first wrote such a criterion down treated their relation as a modeling convenience rather than a question. \citet{HSa_2007} introduce exactly this structure in a filtering problem. One risk-sensitivity operator adjusts continuation values for concerns about model misspecification, conditioned on knowledge of the hidden state.'' A second does so for concerns about misspecification of the distribution of the hidden state,’’ and the two are indexed by separate penalty parameters. Their treatment leaves three points open. First, they single out the case in which the two parameters coincide, where the two operators collapse into one. Second, they adopt the exponential aggregator while noting that any concave increasing function would serve. Third, they record the discipline of the pair as an open problem.

The question matters beyond the filtering problem, for decision theory and for applied work alike. It asks when a smooth ambiguity criterion \citep{KMM_2005} is a second-order expected utility \citep{N_1993, N_2010}, and which aggregator it must then use. Moreover, applied work regularly replaces a two-layer criterion by a single multiplier preference over the predictive distribution --- the object \citet{HM_2018} solve directly and \citet{C_2020} recommends --- and whether that replacement is legitimate depends entirely on how the two parameters are related. 

A simple example makes the two fears concrete. Take an investor choosing among claims on next year's state of the economy, who holds two forecasts she does not fully trust: they agree about the chance of a normal year but disagree about how the remaining probability splits between a boom and a crash. One penalty prices her fear that whichever forecast she uses is itself wrong; the other prices her fear that the weights she puts on the two forecasts are wrong. In the language used below, each forecast is a model, the weights are a second-order prior over models, and the average forecast under that prior is the predictive distribution a Bayesian holding those weights would use.

\begin{figure}[htbp]
    \centering
    \resizebox{\textwidth}{!}{%
    \begin{tikzpicture}[
        box/.style={draw=gray!55, fill=gray!5, rounded corners=4pt, inner sep=7pt, align=center, text width=4.8cm, font=\sffamily\small, execute at begin node={\hyphenpenalty=10000\exhyphenpenalty=10000}},
        centerbox/.style={draw=black!75, fill=gray!14, rounded corners=4pt, inner sep=8pt, align=center, text width=5.2cm, font=\sffamily\small, line width=0.8pt, execute at begin node={\hyphenpenalty=10000\exhyphenpenalty=10000}},
        iff/.style={-{Stealth[scale=1.1]}, thick, shorten >=2pt, shorten <=2pt},
        nest/.style={-{Stealth[scale=1.1]}, semithick, shorten >=2pt, shorten <=2pt},
        noroute/.style={-{Stealth[scale=1.1]}, semithick, dashed, gray!65, shorten >=2pt, shorten <=2pt},
        lab/.style={font=\sffamily\footnotesize, align=center, inner sep=2pt, execute at begin node={\hyphenpenalty=10000\exhyphenpenalty=10000}},
        blocked/.style={circle, draw=gray!65, fill=white, inner sep=0.8pt, font=\sffamily\scriptsize, text=gray!65}
    ]
    \node (soeu) [centerbox] {\textbf{Exponential second-order\\ expected utility}\\[2pt] one entropy penalty,\\ the predictive measure\\ as its baseline\\[2pt] \citep{N_1993, N_2010, S_2011}};
    \node (two) [box, above=2.2cm of soeu] {\textbf{Misspecification-robust\\ smooth ambiguity}\\[2pt] two entropy penalties:\\ one on the models,\\ one on the prior\\[2pt] \citep{HSa_2007, CVHMM_2026}};
    \node (robust) [box, below=2.2cm of soeu] {\textbf{Robust control}\\[2pt] one entropy penalty\\ on a single model\\[2pt] \citep{HS_2001, S_2011}};
    \node (smooth) [box, left=2.2cm of two, text width=4.3cm] {\textbf{Smooth ambiguity}\\[2pt] expected utility inside,\\ smooth aggregator outside\\[2pt] \citep{KMM_2005}};
    \node (ri) [box, right=3.0cm of soeu, text width=4.0cm] {\textbf{Costly information}\\[2pt] relative entropy as the\\ price of an experiment\\[2pt] \citep{S_2003, PST_2023}};
    \draw[iff] (two) -- node[lab, right=2pt] {if and only if\\ the two penalties\\ are matched} (soeu);
    \draw[iff] (robust) -- node[lab, right=2pt] {if and only if\\ the penalty is\\ relative entropy} (soeu);
    \draw[nest] (smooth) -- node[lab, above=2pt, text width=2.1cm] {nested as a boundary case} (two);
    \draw[noroute] (smooth.south) to[out=-90,in=180] node[lab, text=gray!65, pos=0.30, left=2pt] {no route\\ on its own} node[blocked, pos=0.62] {$\times$} (soeu.west);
    \draw[noroute] (ri.west) -- node[lab, above=4pt, text=gray!65, pos=0.5, text width=2.6cm] {shared mathematics, not a shared model} node[blocked, pos=0.5] {$\times$} (soeu.east);
    \end{tikzpicture}%
    }
    \caption{Where the constructions meet. The two solid arrows into the center are if-and-only-if statements rather than sufficient conditions: matching the two penalties is what makes the joint penalty a single relative entropy and collapses the criterion onto the predictive measure (Proposition \ref{prop:chain-rule}, Theorems \ref{thm:collapse} and \ref{thm:mismatch-impossibility}), and the exponential aggregator and relative entropy are each forced (Theorem \ref{thm:necessity}, Proposition \ref{prop:kl-uniqueness}). Smooth ambiguity is the boundary case in which the fear of misspecification vanishes; on its own it has no route to the center, since learning substitutes for the penalty rather than producing it (Proposition \ref{prop:degeneracy}). The cost of information shares the mathematics of the dual but not the decision problem (Proposition \ref{prop:ri-gap}).}
    \label{fig:where-they-meet}
\end{figure}

The obstacle to settling the question is that the two penalties act on different objects. One disciplines beliefs about the state given a model, the other beliefs about which model holds, so a comparison of the two-layer criterion with a single-layer one cannot proceed penalty by penalty; it requires an identity between whole criteria, valid for every model set and every prior rather than in a limit. The route that would make limits available does not help: if the decision-maker learns which model is true and the sample size grows, the criterion converges to a plain expected utility and the penalty parameter disappears from the limit altogether. What is needed instead is an exact identity between the two criteria, and relative entropy turns out to supply one.

In this paper we show that the two layers collapse into a single entropic criterion exactly when one parameter governs both, and that the reason is a property of relative entropy rather than of the exponential form. The chain rule splits a penalty on the joint distribution of model and state into a penalty on the prior and an expected penalty on the conditionals, with the same coefficient attached to each; the two-layer criterion is always the value of a problem on the joint distribution, and matching the two parameters is the only configuration in which the penalty of that problem is a single relative entropy (Proposition \ref{prop:chain-rule}). Under matching the criterion equals the exponential second-order expected utility value for every compact model set and every prior, with the baseline identified as the Bayesian predictive measure and with no limit anywhere (Theorem \ref{thm:collapse}).

Three results turn the collapse into a characterization and settle, in turn, the three points that \citeauthor{HSa_2007} left open. First, matching is not a special case: a mismatched criterion looks entropic to second order, but with an index of caution that blends the two parameters, with weights given by the share of an act’s predictive variance that comes from disagreement among the models (Proposition \ref{prop:effective}). That share varies across acts, so no single parameter and baseline represent such a decision-maker, and a mismatch cannot be repaired by changing the baseline (Theorem \ref{thm:mismatch-impossibility}). Second, the exponential form is not a choice: for a fixed inner parameter, the exponential is the only continuous and strictly increasing aggregator for which the collapse holds (Theorem \ref{thm:necessity}). Third, the pair is disciplined by choice data: the same blending makes the two parameters separately identifiable from acts with different exposure (Corollary \ref{cor:identification}), which answers from the choice side the question \citeauthor{HSa_2007} posed from the statistical side. Finally, the penalty itself is forced: within the power-divergence family, only relative entropy delivers the collapse, the other members departing at the third cumulant (Proposition \ref{prop:kl-uniqueness}).

\subsection*{Related Literature}

This paper sits between a literature that writes two-layer criteria and a literature that studies single-layer ones. On the two-layer side, \citet{HSa_2007} is the direct antecedent, and part (a) of Proposition \ref{prop:chain-rule}, specialized to a common coefficient, is their portmanteau argument transposed to a static setting; \citet{HSa_2022} return to the same structure with one penalty and a set of priors in place of the second penalty, which is a reminder that how to discipline the two layers remains an open modelling choice. \citet{CVHMM_2026} axiomatize the criterion used here and observe, at the close of their Section 6.2, that ``when the two relative entropy penalty parameters are equal, we have the preferences suggested by \citet{C_2020} in his concluding section'', leaving ``a full-fledged analysis of these criteria and of their relationships to future research''. This paper takes up that analysis, and what it contributes is the direction those statements leave open: not that matching suffices, but that nothing else does.

On the single-layer side, the exponential aggregator has two known derivations, and the one given here is a third. \citet[Prop.~2]{KMM_2005} and \citet[\S3.4.1]{S_2011} both reach it from translation invariance, while Theorem \ref{thm:necessity} reaches it from the mixtures that sweep out an open set of acts, which is why the curvature it returns is not a free parameter but the one already disciplining the inner problem; Section \ref{sec:converses} compares the three routes.

The relation to the rational inattention literature is a matter of mathematics shared rather than models shared, and Section \ref{sec:not-routes} makes the limits of the connection precise. The cost that appears in \eqref{eq:wri} is the one that \citet{S_2003} and \citet{PST_2023} study at the level of experiments, and the additivity that pins it down there is the same property of relative entropy that makes the two layers collapse here. It does not follow that the entropic value is the value of an information-acquisition problem, and Proposition \ref{prop:ri-gap} shows that for a fixed act it cannot be; what the dual is instead is the wishful-thinking value of \citet{CL_2019}, reproved by \citet{RSS_2023}, and once a menu of acts is present the entropic value and its dual bound the genuine rationally inattentive value from either side.

\subsection*{Outline}
Section \ref{sec:model} sets up the two-parameter criterion and records the Donsker--Varadhan identity it rests on. Section \ref{sec:results} collects the results. Section \ref{sec:chain} shows what matching is, and Section \ref{sec:collapse} shows that under matching the two layers collapse exactly. Section \ref{sec:converses} proves the two converses on the aggregator side, derives the effective caution parameter of a mismatched decision-maker, and shows that the two parameters are separately identifiable. Section \ref{sec:cost} proves the converse on the cost side, and Section \ref{sec:not-routes} records two routes that do not reach the single-layer target. Section \ref{sec:discussion} states what the characterization delivers and what it does not, and Section \ref{sec:concl} concludes. All proofs are in the appendix.

\section{Model}
\label{sec:model}

\subsection{Primitives}
\label{sec:primitives}
Throughout, $\Omega$ is a \emph{finite} state space and $X$ is a separable metric consequence space, with $\Delta(X)$ its set of Borel probability measures. Let 
\begin{align*}
    B(X) := \{u: X \to \mathbb{R} \mid u \text{ bounded and Borel measurable}\}
\end{align*}
denote the space of bounded utility indices on $X$, and let $\mathcal{F}$ denote the set of all Anscombe--Aumann acts $f: \Omega \to \Delta(X)$. We extend $u\in B(X)$ to $\Delta(X)$ by $u(\ell):=\int_X u\,d\ell$, so that $u(f):=u\circ f\in\mathbb{R}^\Omega$ for every $f\in\mathcal{F}$, and we write
\begin{align*}
    \mathcal{W} \;:=\; \{u(f) \,:\, f\in\mathcal{F}\} \;\subseteq\; \mathbb{R}^\Omega
\end{align*}
for the induced set of state-contingent payoff profiles. 

Write $\underline u:=\inf_X u$ and $\bar u:=\sup_X u$, and assume throughout that $u$ is non-constant, so that $\underline u<\bar u$ (otherwise every result below is vacuous). Mixtures in $\Delta(X)$ sweep out $(\underline u,\bar u)$ pointwise and independently across states, so $\mathcal{W}\supseteq(\underline u,\bar u)^\Omega$; in particular $\mathcal{W}$ has nonempty interior in $\mathbb{R}^\Omega$ and contains the constant profiles $(c,\dots,c)$ for every $c\in(\underline u,\bar u)$, facts used in Theorem \ref{thm:necessity} and Proposition \ref{prop:kl-uniqueness}. Boundedness of $u$ makes $\mathbb{E}_p[u(f)]$ and $\mathbb{E}_q[\exp(\pm u(f)/\theta)]$ finite for every $p,q \in \Delta(\Omega)$, every $f\in\mathcal{F}$ and every $\theta > 0$. 

Restricting $\Omega$ to be finite keeps the exposition self-contained; for most results the restriction is only a convenience. Lemma \ref{lem:dv} and Theorem \ref{thm:collapse} extend to a compact Polish $\Omega$ through the apparatus of \citet{DE_1997}. Theorems \ref{thm:necessity} and \ref{thm:mismatch-impossibility} are stated for finite $\Omega$ because their proofs sweep out an open set of acts by linear algebra on $\mathbb{R}^{|\Omega|}$, and Remark \ref{rem:binary-richness} extends them to an arbitrary measurable state space under a mild richness condition. Proposition \ref{prop:degeneracy} is the genuine exception, where finiteness is essential rather than cosmetic (Appendix \ref{app:degeneracy}). 

Let $\Delta_{++}(\Omega) \subset \Delta(\Omega)$ denote the relative interior of the simplex (full-support distributions). Whenever $D_{\mathrm{KL}}(p\|q)$ appears below, $p, q \in \Delta_{++}(\Omega)$, so that $p \ll q$ automatically and $D_{\mathrm{KL}}(p\|q) < \infty$. The set of first-order models $\Pi \subset \Delta_{++}(\Omega)$ is compact (the role of \citeauthor{CVHMM_2026}'s (\citeyear{CVHMM_2026}) set $Q$ of structured models) and second-order priors $\mu \in \Delta(\Pi)$ are Borel probability measures. The set $\Pi$ need neither be convex nor contain the baseline $q$: Section \ref{sec:two-layers} shows that $q$ emerges from $(\Pi,\mu)$ as the induced predictive measure.

\subsection{Two Layers}
\label{sec:two-layers}

A decision-maker who entertains several statistical models faces uncertainty of two kinds, and the criterion we study prices them separately. Within a model $\pi$ she may fear that the model itself is wrong, and the standard device for that fear is a relative-entropy penalty on departures from $\pi$ \citep{HS_2001, S_2011}: the model is evaluated not by $\mathbb{E}_\pi[u(f)]$ but by the robust-control value
\begin{equation}
    v^\lambda_\pi(f) := \min_{p\in\Delta_{++}(\Omega)} \big\{ \mathbb{E}_p[u(f)] + \lambda\, D_{\mathrm{KL}}(p\|\pi) \big\}.
    \label{eq:inner}
\end{equation}
Across models she faces a second kind of uncertainty, which model is the right one, and the standard device for that is a smooth aggregator applied to the values $v^\lambda_\pi(f)$ under a second-order prior $\mu$. Writing $\phi_\xi(x):=-\exp(-x/\xi)$ for the exponential aggregator with parameter $\xi$, the criterion of \citet{CVHMM_2026} is
\begin{equation}
    V^{\lambda,\xi}_\Pi(f) := \phi_\xi^{-1}\Big(\mathbb{E}_\mu\big[\phi_\xi\big(v^\lambda_\pi(f)\big)\big]\Big) = -\xi\log\mathbb{E}_\mu\big[\exp\big(-v^\lambda_\pi(f)/\xi\big)\big].
    \label{eq:two-layer}
\end{equation}
We say that the curvature is \emph{matched} when $\lambda=\xi$, and write $\theta$ for the common value and $V^{\theta,\theta}_\Pi$ for the resulting criterion. The two parameters are not interchangeable: $\lambda$ measures distrust of any given model and $\xi$ measures distrust of the model set, and \citeauthor{CVHMM_2026} describe them in exactly those terms.\footnote{In their words, ``the parameter $\xi>0$ captures aversion to prior uncertainty, while the parameter $\lambda>0$ is a fear of model misspecification index''. Each is an entropy penalty parameter rather than a curvature, so the associated Arrow--Pratt index is $1/\xi$. Equations \eqref{eq:inner}--\eqref{eq:two-layer} are their criterion (33). We use their Section 6.1 formulation, in which the second-order prior is a single measure $\mu$; their Section 6.2 variant replaces $\mu$ by a minimization over priors, which for the entropic penalty coincides with it by Donsker--Varadhan duality at the prior level.} The criterion nests plain smooth ambiguity as $\lambda\to\infty$, since $v^\lambda_\pi(f)\to\mathbb{E}_\pi[u(f)]$, and \citeauthor{GS_1989}'s (\citeyear{GS_1989}) max-min criterion as $\xi\to0$.

The same two-parameter structure appears in the applied literature that introduced it. In the hidden-state problem of \citet{HSa_2007}, quoted in the introduction, two risk-sensitivity operators are indexed by penalty parameters $(\theta_1,\theta_2)$: their $\theta_1$ plays the role of our $\lambda$ and their $\theta_2$ the role of our $\xi$, with the hidden state in place of the model index. Three features of their treatment organize what follows. First, they single out the case $\theta_1=\theta_2$, in which the two operators collapse into one. Second, they adopt the exponential aggregator as a convenience, noting that any concave increasing $\psi$ would do. Third, they record the discipline of the pair as an open problem.

\subsection{The Single-Layer Target}
\label{sec:target}

The object we compare \eqref{eq:two-layer} with is the exponential second-order expected utility value
\begin{equation}
    V(f) := -\theta\log\mathbb{E}_q\big[\exp(-u(f)/\theta)\big],
    \label{eq:anchor}
\end{equation}
which is a single-layer criterion in the sense that it applies one entropy penalty to one baseline $q$. Its two readings, as a second-order expected utility with an exponential aggregator and as a multiplier preference, are the two sides of the Donsker--Varadhan identity, which we record once and use throughout.

\begin{lemma}[Donsker--Varadhan Variational Identity]
\label{lem:dv}
For every $\theta>0$, every $q\in\Delta_{++}(\Omega)$ and every $f\in\mathcal{F}$,
\begin{equation}
    -\theta \log \mathbb{E}_q\!\left[\exp(-u(f)/\theta)\right] = \min_{p \in \Delta_{++}(\Omega)} \left\{ \mathbb{E}_p[u(f)] + \theta D_{\mathrm{KL}}(p \| q) \right\},
    \label{eq:dv}
\end{equation}
with the minimum attained uniquely at the Gibbs measure
\begin{align*}
    p^*(\omega) \;=\; \frac{q(\omega)\exp(-u(f(\omega))/\theta)}{\mathbb{E}_q[\exp(-u(f)/\theta)]}.
\end{align*}
\end{lemma}
With the convention $0\log 0=0$, both the identity and the uniqueness of the minimizer persist when the minimization in \eqref{eq:anchor} is taken over all of $\Delta(\Omega)$: since $q$ is strictly positive, $D_{\mathrm{KL}}(p\|q)<\infty$ for every $p\in\Delta(\Omega)$, and Gibbs' inequality is strict off $p^*$. The argument uses $u(f)$ only as a vector in $\mathbb{R}^\Omega$, so \eqref{eq:anchor} and the Gibbs form of the minimizer hold verbatim with $u(f)$ replaced by any payoff profile $w\in\mathbb{R}^\Omega$. This is the form used for $W_{\mathrm{RI}}$ in Section \ref{sec:ri-gap} and in Section \ref{sec:discussion}.

\begin{proof}
This is the Donsker--Varadhan variational formula; see \citet[Prop.~1.4.2 and Thm.~1.2.1]{DE_1997}, or, in the present decision-theoretic notation, equation (8) of \citet[Sec.~3.3]{S_2011}, where it is stated as ``the variational formula'' and attributed to \citeauthor{DE_1997}. For completeness, for any $p\ll q$,
\begin{align*}
    \mathbb{E}_p[u(f)] + \theta D_{\mathrm{KL}}(p\|q)
    - \Big(\!-\theta\log\mathbb{E}_q[\exp(-u(f)/\theta)]\Big)
    \;=\; \theta D_{\mathrm{KL}}(p\|p^*) \;\ge\; 0,
\end{align*}
with equality if and only if $p=p^*$, where $p^*$ is as displayed above; this is verified by direct substitution using $\log(dp^*/dq) = -u(f)/\theta - \log \mathbb{E}_q[\exp(-u(f)/\theta)]$.
\end{proof}

Read from right to left, Lemma \ref{lem:dv} says that $V$ is the multiplier preference of \citet{HS_2001} with reference measure $q$ and penalty $\theta$; read from left to right, it says that $V$ is a second-order expected utility with the exponential aggregator $\phi_\theta$. Which of the two readings is the useful one depends on the question, and the results below use both.

\section{Results}
\label{sec:results}

Section \ref{sec:chain} says what matching is, Section \ref{sec:collapse} what it does, Sections \ref{sec:converses} and \ref{sec:cost} why nothing else does it, and Section \ref{sec:not-routes} what it is not.

\subsection{Matching as a Single Joint Penalty}
\label{sec:chain}

Before asking when the two layers collapse, we examine what matching \emph{is}. The answer is a property of relative entropy rather than of the exponential form: relative entropy is the divergence that splits additively when a joint distribution is factored into a marginal and a conditional, and a single penalty on the joint distribution therefore delivers one coefficient to each layer. Let $\bar P$ denote the joint baseline on $\Pi\times\Omega$ defined by $\bar P(d\pi,d\omega):=\mu(d\pi)\,\pi(d\omega)$, whose $\Pi$-marginal is $\mu$, whose $\Omega$-conditional given $\pi$ is $\pi$ itself, and whose $\Omega$-marginal is the predictive measure $\bar\pi:=\int_\Pi\pi\,d\mu(\pi)$; note that $\bar P$ is not a product measure. The chain rule for relative entropy gives
\begin{equation}
    D_{\mathrm{KL}}(P\|\bar P) \;=\; D_{\mathrm{KL}}(\nu\|\mu) + \int_\Pi D_{\mathrm{KL}}\big(p_\pi\,\|\,\pi\big)\,d\nu(\pi),
    \label{eq:chain-rule}
\end{equation}
where $\nu\ll\mu$ is the $\Pi$-marginal of $P$ and $(p_\pi)$ a version of its conditionals, so that a decision-maker who penalizes $D_{\mathrm{KL}}(P\|\bar P)$ once, with one coefficient, necessarily applies that same coefficient to both terms. The two-layer criterion, by contrast, carries one coefficient per term. For $P\ll\bar P$ write
\begin{equation}
    C_{\lambda,\xi}(P) \;:=\; \xi\,D_{\mathrm{KL}}(\nu\|\mu) + \lambda\int_\Pi D_{\mathrm{KL}}\big(p_\pi\,\|\,\pi\big)\,d\nu(\pi)
    \label{eq:two-coefficient-penalty}
\end{equation}
for the penalty that charges the prior at rate $\xi$ and the conditionals at rate $\lambda$; by \eqref{eq:chain-rule}, $C_{\theta,\theta}=\theta D_{\mathrm{KL}}(\cdot\|\bar P)$. The next result says that the two-layer criterion is always the value of a problem penalizing $C_{\lambda,\xi}$ on the joint distribution, and that this penalty is a single relative entropy exactly under matching.

\begin{proposition}[Matching is a single joint penalty]
\label{prop:chain-rule}
Let $\Pi\subset\Delta_{++}(\Omega)$ be compact and let $\mu\in\Delta(\Pi)$.
\begin{enumerate}
    \item[(a)] For every $\lambda,\xi>0$ and every $f\in\mathcal{F}$, the two-layer criterion \eqref{eq:two-layer} is the value of a problem on the joint distribution of model and state,
    \begin{equation}
        V^{\lambda,\xi}_\Pi(f) \;=\; \min_{P\ll\bar P}\Big\{\mathbb{E}_P[u(f)] + C_{\lambda,\xi}(P)\Big\},
        \label{eq:joint-penalty}
    \end{equation}
    the minimum being attained at the distribution whose $\Pi$-marginal is the Gibbs tilt of $\mu$ by $\pi\mapsto v^\lambda_\pi(f)$ at rate $\xi$ and whose conditional given $\pi$ is the Gibbs tilt of $\pi$ by $u(f)$ at rate $\lambda$.
    \item[(b)] Suppose $\mu$ is not a point mass. The penalty $C_{\lambda,\xi}$ is a single relative entropy, in the sense that
    \begin{equation}
        C_{\lambda,\xi}(P) \;=\; \theta\,D_{\mathrm{KL}}(P\,\|\,\bar P')
        \qquad\text{for every } P\ll\bar P,
        \label{eq:single-penalty}
    \end{equation}
    for some $\theta>0$ and some probability measure $\bar P'$ on $\Pi\times\Omega$, if and only if $\lambda=\xi$; in that case $\bar P'=\bar P$ and $\theta=\lambda=\xi$.
\end{enumerate}
In particular, under matching $\lambda=\xi=\theta$ the criterion is the value of the problem penalizing the single relative entropy $\theta D_{\mathrm{KL}}(\cdot\|\bar P)$.
\end{proposition}

\begin{proof}[Proof sketch]
For (a), minimize \eqref{eq:joint-penalty} in two stages: over the conditionals, Lemma \ref{lem:dv} at each $\pi$ with parameter $\lambda$; over the $\Pi$-marginal, Lemma \ref{lem:dv} on $\Pi$ with parameter $\xi$, applied to $\pi\mapsto v^\lambda_\pi(f)$. The ``if'' half of (b) is the chain rule \eqref{eq:chain-rule}. For the ``only if'' half, evaluating \eqref{eq:single-penalty} at $P=\bar P$ forces $\bar P'=\bar P$, after which holding the conditionals at the models forces $\theta=\xi$, because $\mu$ is not a point mass, and holding the prior at $\mu$ forces $\theta=\lambda$. Appendix \ref{app:pf-chain} gives the details.
\end{proof}

The sufficiency direction is not new, and its provenance matters for what the proposition is for. \citet[\S6.1]{HSa_2007} observe, in their filtering problem, that when a common value of $\theta$ appears in the two operators their sequential application ``can be replaced by a single operator'', which is ``the outcome of a portmanteau minimization problem over a single relative density''. What Proposition \ref{prop:chain-rule} adds is the converse, at the level of the penalty itself. The two-layer criterion is always the value of a joint-law problem, by part (a); matching is what makes the penalty of that problem a single relative entropy, and part (b) says it is the only configuration that does. Whether a modeller may treat the two layers as expressions of one attitude therefore has a sharp answer.

The proviso that $\mu$ not be a point mass cannot be dropped, and its failure is instructive rather than awkward. When $\mu=\delta_\pi$ the outer layer is vacuous, $\xi$ is not identified at all, and $V^{\lambda,\xi}_\Pi=v^\lambda_\pi$ for every $\xi$; there is nothing to match because there is no second layer. Everything below therefore concerns model sets over which the prior is genuinely dispersed, which is also the only case in which the two layers describe different things.

\begin{remark}[Penalties versus values]
\label{rem:penalty-vs-value}
Proposition \ref{prop:chain-rule}(b) is a statement about the penalty $C_{\lambda,\xi}$ as a functional on joint distributions, not about the value $V^{\lambda,\xi}_\Pi(f)$. A separate question is whether, for a given $(\Pi,\mu)$ with $\lambda\ne\xi$, the value could nonetheless coincide for every act with the value of \emph{some} single-penalty problem $\min_P\{\mathbb{E}_P[u(f)]+\theta' D_{\mathrm{KL}}(P\|\bar P)\}$. By Lemma \ref{lem:dv} and Theorem \ref{thm:collapse} below, the latter value is $-\theta'\log\mathbb{E}_{\bar\pi}[\exp(-u(f)/\theta')]$, so this is the single-layer reduction question of Theorem \ref{thm:mismatch-impossibility} with the baseline fixed at $\bar\pi$; that theorem rules it out whenever exposure varies across acts, and in the knife-edge case for the configuration exhibited there.
\end{remark}

\subsection{Collapse under Matched Parameters}
\label{sec:collapse}

An example shows the collapse before the theorem explains it. An investor entertains two forecasts of next period's state, $\Omega=\{\text{boom},\text{normal},\text{crash}\}$, one optimistic and one pessimistic, and weights them equally:
\begin{center}
\begin{tabular}{lccc}
\toprule
 & boom & normal & crash \\
\midrule
$\pi_1$ (optimistic) & $0.50$ & $0.30$ & $0.20$ \\
$\pi_2$ (pessimistic) & $0.20$ & $0.30$ & $0.50$ \\
\midrule
$\bar\pi$ (predictive, $\mu=(\tfrac12,\tfrac12)$) & $0.35$ & $0.30$ & $0.35$ \\
\bottomrule
\end{tabular}
\end{center}
The two forecasts agree exactly on the probability of the normal state and disagree about how the remaining mass is split between boom and crash. Consider four claims, each paying one util in the indicated event: $x_N$ in the normal state, $x_B$ in boom, $x_C$ in crash, and $x_{BC}$ in boom or crash.

\begin{example}[The collapse under matching]
\label{ex:collapse}
Take $\lambda=\xi=\theta=1$. Evaluating each claim in two stages through \eqref{eq:two-layer}, and in one stage through \eqref{eq:anchor} with baseline $\bar\pi$, gives
\begin{center}
\begin{tabular}{lcc}
\toprule
claim & two-layer $V^{\theta,\theta}_\Pi$ & one-layer $V$ at $q=\bar\pi$ \\
\midrule
$x_N$ & $0.2102719564$ & $0.2102719564$ \\
$x_B$ & $0.2500551872$ & $0.2500551872$ \\
$x_C$ & $0.2500551872$ & $0.2500551872$ \\
$x_{BC}$ & $0.5842647782$ & $0.5842647782$ \\
\bottomrule
\end{tabular}
\end{center}
The two columns agree to every digit displayed, and the agreement is exact rather than numerical.
\end{example}

The reason is visible in Lemma \ref{lem:dv}. The inner problem returns the entropic value $v^\theta_\pi(f)$ in the log-sum-exp form \eqref{eq:anchor}, and the outer aggregator $\phi_\theta$ is the exact functional inverse of the map $x\mapsto-\theta\log(-x)$ that produced it. The outer average therefore passes through the two logarithms and lands on a mixture of the \emph{measures} rather than remaining a nonlinear function of the \emph{values}. Fubini then replaces $\int_\Pi\int_\Omega$ by $\int_\Omega$ against $\bar\pi$, and nothing survives of the model set except its predictive measure.

\begin{theorem}[The two layers collapse under matching]
\label{thm:collapse}\label{lem:predictive-reduction}
Let $\Pi\subset\Delta_{++}(\Omega)$ be compact and $\mu\in\Delta(\Pi)$. Suppose the curvature is matched, $\lambda=\xi=\theta$. Then, for every act $f \in \mathcal{F}$,
\begin{equation}
    V_\Pi(f) = -\theta\log\mathbb{E}_{\bar\pi}\big[\exp(-u(f)/\theta)\big] = V(f), \qquad \bar\pi(\cdot) := \int_\Pi \pi(\cdot)\,d\mu(\pi).
    \label{eq:predictive-reduction}
\end{equation}
\end{theorem}
\begin{proof}
See Appendix \ref{app:pf-collapse}. The displayed identity is an application of Lemma \ref{lem:dv} at each $\pi$, followed by the cancellation described above and Fubini--Tonelli.
\end{proof}

\begin{remark}
\label{rem:matched-theta}
Theorem \ref{thm:collapse} does not require $\bar\pi\in\Pi$, a point that is easy to read past. The simplex $\Delta(\Omega)$ is convex but $\Pi$ need not be, so the predictive measure can lie strictly outside the set of models the decision-maker entertains --- the familiar situation of a de Finetti mixture. Read $(\Pi,\mu)$ as the mixing representation of an exchangeable sequence in the sense of \citet{HS_1955}. Then $\bar\pi$ is its one-step-ahead predictive distribution. The mixture $\int_\Pi\pi^{\otimes\infty}\,d\mu(\pi)$ equals the i.i.d.\ law $\bar\pi^{\otimes\infty}$ only when $\mu$ is a point mass, but the one-shot evaluation of an act sees only $\bar\pi$. Theorem \ref{thm:collapse}(b) identifies the baseline $q$ with exactly this object, in the predictive spirit of \citet{DP_2022}.
\end{remark}

No sample size and no limit appears anywhere in the argument, and this is what distinguishes Theorem \ref{thm:collapse} from the large-sample results that usually motivate predictive criteria. Section \ref{sec:not-routes} records the contrast formally: if the decision-maker instead learns which model is true and the sample size grows, the criterion converges to a plain expected utility and $\theta$ disappears. The collapse here is algebraic, holds at every sample size including none, and leaves $\theta$ in place.

\subsection{The Exponential Aggregator and Mismatched Parameters}
\label{sec:converses}

Theorem \ref{thm:collapse} fixes the aggregator and asks what matching buys; the first converse fixes the matching and asks what the aggregator must be. \citet{HSa_2007} write that one ``could use $\tilde R(\hat V)=\psi^{-1}(\mathbb{E}[\psi(\hat V)])$ for some concave increasing function $\psi$'', and adopt the exponential ``because of its explicit link to entropy-based robustness''. The next result says that the link is not a matter of taste: at a fixed inner curvature, the exponential is the only aggregator for which the collapse holds at all.

\begin{theorem}[Necessity of the exponential aggregator]
\label{thm:necessity}
Let $|\Omega|\ge2$ and fix $\theta>0$. Let $\phi:\mathbb{R}\to\mathbb{R}$ be continuous and strictly increasing, and let
\begin{align*}
    V^\phi_\Pi(f) \;:=\; \phi^{-1}\Big(\mathbb{E}_\mu\big[\phi\big(v^\theta_\pi(f)\big)\big]\Big)
\end{align*}
denote the criterion \eqref{eq:two-layer} with $\phi$ in the outer aggregation. If $V^\phi_\Pi(f)=V(f)$, with $V$ evaluated at the predictive baseline $q=\bar\pi$, for every compact $\Pi\subset\Delta_{++}(\Omega)$, every $\mu\in\Delta(\Pi)$ and every $f\in\mathcal{F}$, then there are $a<0$ and $b\in\mathbb{R}$ such that $\phi(x)=a\exp(-x/\theta)+b$ for every $x\in(\underline u,\bar u)$.
\end{theorem}

\begin{proof}
See Appendix \ref{app:pf-necessity}. Instantiating the hypothesis at a two-model set turns the collapse into an equality of quasi-arithmetic means on an open square of acts, and the arbitrariness of the prior weights then forces affinity through Lemma \ref{lem:affine-rigidity}, with no regularity beyond what is assumed here.
\end{proof}

\begin{remark}[Relation to \citeauthor{KMM_2005}'s (\citeyear{KMM_2005}) Proposition 2]
\label{rem:kmm-prop2}
The exponential form has two familiar derivations, and Theorem \ref{thm:necessity} reaches it by a third. \citet[Prop.~2]{KMM_2005} obtain it, without assuming differentiability, from \emph{constant ambiguity attitude} --- invariance of the ranking under a common constant shift in utility in every state. \citet[\S3.4.1]{S_2011} obtains it from the requirement that a variational preference also be a second-order expected utility one --- multiplier preferences are precisely that intersection --- with $\theta$ inherited from the robustness parameter. Both routes rest on translation invariance and lead to the generalized Pexider equation $\phi(x+k)=\alpha(k)\phi(x)+\beta(k)$ \citep[p.~58]{S_2011}, treated on a restricted open domain by \citet{A_2005} and \citet{GNA_2005}. Theorem \ref{thm:necessity} rests instead on the mixture structure of the domain of acts. It imposes no axiom on attitudes; it requires only that the two-layer criterion \eqref{eq:two-layer} agree with \emph{some} single-layer entropic value on the open set of acts that a two-point model set sweeps out. The resulting equation is the equality of quasi-arithmetic means (Remark \ref{rem:affine-rigidity-lit}). The curvature is pinned down differently too: their $\alpha$ is free, identified by ambiguity attitude, whereas here $\theta$ is the parameter that already disciplines the inner robustification \eqref{eq:inner}.
\end{remark}

A second escape route remains open. A decision-maker whose curvature is genuinely mismatched might still be representable in single-layer entropic form, if she is allowed a different target curvature $\theta'$ and a different baseline $q$; nothing in Theorem \ref{thm:necessity} rules this out, because that theorem holds the inner curvature fixed and varies only the aggregator. Closing this route requires knowing what a mismatched criterion looks like from the outside, and the answer turns out to be interpretable rather than merely negative.

\begin{proposition}[Effective caution parameter]
\label{prop:effective}
Let $\lambda,\xi>0$, let $\Pi\subset\Delta_{++}(\Omega)$ be compact and let $\mu\in\Delta(\Pi)$. For a payoff profile $w\in\mathbb{R}^\Omega$ write
\begin{align*}
    W(w):=\mathbb{E}_\mu\big[\mathrm{Var}_\pi(w)\big], \qquad B(w):=\mathrm{Var}_\mu\big(\mathbb{E}_\pi[w]\big), \qquad R(w):=\frac{B(w)}{W(w)+B(w)},
\end{align*}
the within-model and between-model components of the predictive variance, so that $\mathrm{Var}_{\bar\pi}(w)=W(w)+B(w)$ and $R(w)\in[0,1]$ whenever $w$ is non-constant. Then, as $t\to0$,
\begin{align*}
    V^{\lambda,\xi}_\Pi(c+tw) \;=\; c + t\,\mathbb{E}_{\bar\pi}[w] \;-\; \frac{t^2}{2}\left[\frac{W(w)}{\lambda}+\frac{B(w)}{\xi}\right] + O(t^3),
\end{align*}
and a single-layer value $-\theta'\log\mathbb{E}_q[\exp(-(c+tw)/\theta')]$ agrees with it through order $t^2$ if and only if $q=\bar\pi$ and
\begin{equation}
    \frac{1}{\theta'} \;=\; \big(1-R(w)\big)\frac{1}{\lambda} \;+\; R(w)\,\frac{1}{\xi}.
    \label{eq:effective-curvature}
\end{equation}
\end{proposition}

\begin{proof}
See Appendix \ref{app:pf-effective}. Expand both layers in $t$ and use the decomposition of the predictive variance into its within-model and between-model parts.
\end{proof}

Equation \eqref{eq:effective-curvature} says that the index of caution a mismatched decision-maker appears to use, $1/\theta'$, is the variance-weighted average of the two indices she actually uses, and the weight is a property of the act. Call $R(w)$ the act's \emph{exposure to disagreement}: it is the share of the predictive variance of the payoff that comes from the models disagreeing about its mean rather than from variability within each model. An act on which the models agree has $R=0$ and reveals $\lambda$ alone; an act whose entire predictive variance is disagreement has $R=1$ and reveals $\xi$ alone; and in between the two parameters blend. Under matching the weight is irrelevant, since both ends of the average are the same number, which is Theorem \ref{thm:collapse} seen from the second order.

\begin{example}[Exposure in the example]
\label{ex:exposure}
In the environment of Example \ref{ex:collapse} the two forecasts agree on the normal state, so $x_N$ has $R=0$, and they disagree about boom and crash, so $x_B$ and $x_C$ have $R=0.0989$. The claim $x_{BC}$, which pays unless the normal state occurs, has $R=0$ as well, because the two forecasts agree that the probability of boom or crash is $0.7$: exposure is a property of the act, not of the states its payoff happens to involve. Taking $\lambda=1$ and $\xi=3$ and solving for the curvature that reproduces each claim's two-layer value exactly gives
\begin{center}
\begin{tabular}{lccc}
\toprule
claim & exposure $R$ & $\theta'$ from \eqref{eq:effective-curvature} & $\theta'$ required exactly \\
\midrule
$x_N$ & $0$ & $1.000000$ & $1.000000$ \\
$x_B$ & $0.0989$ & $1.070588$ & $1.061654$ \\
$x_C$ & $0.0989$ & $1.070588$ & $1.061654$ \\
$x_{BC}$ & $0$ & $1.000000$ & $1.000000$ \\
\bottomrule
\end{tabular}
\end{center}
Two claims in the same environment demand two different curvatures, so no single pair $(\theta',q)$ represents this decision-maker.
\end{example}

The example is an instance of a general obstruction, and the obstruction is exactly the act-dependence of exposure.

\begin{theorem}[No single-layer reduction under mismatch]
\label{thm:mismatch-impossibility}
Let $\lambda\ne\xi$. Suppose $V^{\lambda,\xi}_\Pi(f)=-\theta'\log\mathbb{E}_q[\exp(-u(f)/\theta')]$ for every $f\in\mathcal{F}$. Then $q=\bar\pi$ and $R(\cdot)$ is constant on $\mathcal{W}$. Consequently:
\begin{enumerate}
    \item[(i)] if $\mathcal{W}$ contains two profiles with different exposure, no such pair $(\theta',q)$ exists; this is the case whenever $\mathrm{Cov}_\mu(\pi)$ is not proportional to the predictive covariance matrix $\mathrm{diag}(\bar\pi)-\bar\pi\bar\pi^{\top}$, and in particular whenever $|\Omega|\ge3$ and $\Pi$ contains exactly two distinct models;
    \item[(ii)] if exposure is constant, as it is when $|\Omega|=2$, the first two orders cannot separate the two criteria, but the reduction can still fail at the third: there are a two-point model set and a prior, with $|\Omega|=2$, for which no such pair $(\theta',q)$ exists, the obstruction appearing at the third cumulant.
\end{enumerate}
\end{theorem}

\begin{proof}
See Appendix \ref{app:pf-mismatch}. Exact equality makes both sides real-analytic functions of the scale of the payoff profile, so every Taylor coefficient must agree; the first two orders give Proposition \ref{prop:effective}, which proves the opening claim and (i), and the third order gives (ii).
\end{proof}

Part (i) is the generic case and its proof is two lines once Proposition \ref{prop:effective} is available; part (ii) covers the knife-edge in which the model covariance happens to be a multiple of the predictive covariance, where the first two orders cannot distinguish the two criteria. The knife-edge is not pathological. It holds for two states, where an act has only one direction to vary in, and it holds for the conjugate Bayesian benchmark.

\begin{remark}[The Dirichlet case]
\label{rem:dirichlet}
If the second-order prior is Dirichlet with parameter $\alpha$, then $\mathrm{Cov}_\mu(\pi)=\big(\mathrm{diag}(\bar\pi)-\bar\pi\bar\pi^{\top}\big)/(\alpha_0+1)$ with $\alpha_0=\sum_i\alpha_i$, so exposure is the same for every act and equals $R\equiv1/(\alpha_0+1)$. The effective caution parameter is then a single number, and $\alpha_0$ has the usual reading as the strength of the prior: the more the decision-maker already believes she knows about which model holds, the smaller the exposure of any act and the closer her apparent caution to $\lambda$.
\end{remark}

Because exposure is observable in principle, the mismatch that Theorem \ref{thm:mismatch-impossibility} rules out is also measurable.

\begin{corollary}[Identification]
\label{cor:identification}
If two payoff profiles have different exposure and the curvature that each requires is observed, then $\lambda$ and $\xi$ are separately identified by \eqref{eq:effective-curvature}.
\end{corollary}

\citet[fn.~22]{HSa_2007} record the discipline of the pair $(\theta_1,\theta_2)$ as an open problem, and propose approaching it through detection-error probabilities, that is, from the statistical side. Corollary \ref{cor:identification} answers the same question from the choice side: vary how much of an act's predictive variance comes from disagreement among the models, and the two parameters separate. Whether the variation needed is available in any particular application is an empirical matter, and Section \ref{sec:discussion} returns to it.

Exposure has a direct reading in applications. A claim whose risk is mostly within-model, such as a short-dated claim on an aggregate about which the forecasts broadly agree, is priced almost entirely at the rate $\lambda$ that answers to misspecification. A claim whose risk is disagreement among the models, such as a tail layer or a long-horizon payoff, is priced increasingly at the rate $\xi$ that answers to uncertainty about the models. The caution an observer reads off the decision-maker's choices therefore changes with the claim she is choosing over.

\begin{remark}[Beyond a Finite State Space]
\label{rem:binary-richness}
Theorems \ref{thm:necessity} and \ref{thm:mismatch-impossibility} are stated for finite $\Omega$ because their proofs use linear algebra on $\mathbb{R}^{|\Omega|}$ to sweep out an open set of acts. Both extend to an arbitrary measurable state space $(\Omega,\mathcal{F})$ under a mild richness condition on the models the decision-maker entertains. Say that the model set is \emph{binary rich} if there are an event $A\in\mathcal{F}$ and two admissible models $\pi_1,\pi_2$ with
\begin{align*}
    \pi_1(A),\,\pi_2(A)\in(0,1)
    \qquad\text{and}\qquad
    \pi_1(A)\ne\pi_2(A).
\end{align*}
For finite $\Omega$ with $|\Omega|\ge2$ and $\Pi\subset\Delta_{++}(\Omega)$ this is automatic as soon as two distinct models are available: take $A=\{\omega\}$ for any state at which they differ. Under binary richness Theorem \ref{thm:necessity} holds verbatim, and Theorem \ref{thm:mismatch-impossibility} holds in the uniform form stated as its second conclusion.

Appendix \ref{app:binary-richness} carries out the reduction: restricting attention to acts measurable with respect to $\mathcal{G}:=\{\emptyset,A,A^c,\Omega\}$ turns the hypothesis of either theorem into the corresponding hypothesis on the two-point state space $\{A,A^c\}$, to which the proofs above apply verbatim. One qualification, also recorded there, applies to Theorem \ref{thm:mismatch-impossibility}. Binary richness supplies a single pair of models, whereas the proof of that theorem \emph{selects} the pair $\pi_1(\omega_1)=\tfrac34$, $\pi_2(\omega_1)=\tfrac14$ to make the third-order obstruction explicit. For given models, what survives verbatim is the uniform conclusion, and it survives at second order alone.
\end{remark}

\subsection{The Uniqueness of Relative Entropy}
\label{sec:cost}

Matching disciplines the aggregator, and the two converses of Section \ref{sec:converses} say that nothing else aggregates this way. The remaining freedom is in the inner problem: \eqref{eq:inner} measures departures from a model by relative entropy, and one may ask what happens if some other divergence is used. At the level of preferences the answer is known in considerable generality, since \citet[\S3.4.1]{S_2011} observes that a variational preference whose cost is a statistical distance other than relative entropy has no second-order expected utility representation for any aggregator, and conversely. What the result below adds is quantitative: within a one-parameter family that contains the familiar alternatives, it locates the order at which the collapse fails.

For the remainder of Section \ref{sec:cost} we work directly with payoff profiles $w\in\mathcal{W}\subseteq\mathbb{R}^\Omega$ rather than with the underlying acts, writing $\bar w:=\mathbb{E}_q[w]$, $\kappa_2:=\mathrm{Var}_q(w)$ and $\kappa_3:=\mathbb{E}_q[(w-\bar w)^3]$ for the first three cumulants of $w$ under $q$. Consider the power-divergence (Cressie--Read) family
\begin{equation}
    D_\gamma(p\|q) := \sum_\omega q(\omega)\,f_\gamma\!\big(p(\omega)/q(\omega)\big),
    \qquad
    f_\gamma(t) := \frac{t^\gamma - \gamma t + \gamma - 1}{\gamma(\gamma-1)}\ \ (\gamma\ne0,1),
    \label{eq:power-divergence}
\end{equation}
with $f_1(t):=t\log t-t+1$ and $f_0(t):=t-1-\log t$ the $\gamma\to1$ and $\gamma\to0$ limits, so that $D_1(p\|q)=D_{\mathrm{KL}}(p\|q)$ and $D_0(p\|q)=D_{\mathrm{KL}}(q\|p)$ is reverse KL. The family varies continuously in $\gamma$, connects to Rényi-type divergences, and contains the ordinary $\chi^2$ divergence $D_{\chi^2}(p\|q):=\sum_\omega(p(\omega)-q(\omega))^2/q(\omega)$ at $\gamma=2$, up to a factor of two ($D_{\chi^2}=2D_2$). Write
\begin{align*}
    \Lambda_\gamma(w;\theta) \;:=\; \min_{p\in\Delta_{++}(\Omega)}\big\{\mathbb{E}_p[w] + \theta D_\gamma(p\|q)\big\},
\end{align*}
suppressing $\theta$ when it is clear from context. The first lemma records the shape of the optimal tilt; the proposition then shows that only $\gamma=1$ produces a value function of log-sum-exp form.

\begin{lemma}[Optimal Tilts in the Power-Divergence Family]
\label{lem:power-tilt}
Fix $q\in\Delta_{++}(\Omega)$, $w\in\mathcal{W}$ and $\gamma\in\mathbb{R}$, and let $\theta>0$ satisfy $\theta>(\gamma-1)\,\mathrm{osc}(w)$ when $\gamma>1$, where $\mathrm{osc}(w):=\max_\omega w(\omega)-\min_\omega w(\omega)$; for $\gamma\le1$ no restriction on $\theta$ is needed. Then the minimum defining $\Lambda_\gamma(w)$ is attained uniquely, at
\begin{align*}
    p^*(\omega) &\ \propto\ q(\omega)\exp\big(-w(\omega)/\theta\big) && (\gamma=1), \\
    p^*(\omega) &\ \propto\ q(\omega)\big[1+(\gamma-1)(\nu-w(\omega))/\theta\big]^{1/(\gamma-1)} && (\gamma\ne1),
\end{align*}
where $\nu$ denotes the multiplier on the constraint $\mathbb{E}_q[p^*/q]=1$. The tilt is exponential in $w$ if and only if $\gamma=1$; for every other $\gamma$ it is a \emph{power} tilt.
\end{lemma}
\begin{proof}
See Appendix \ref{app:pf-power-tilt}.
\end{proof}

\begin{proposition}[KL Is the Unique Power Divergence with a Log-Sum-Exp Dual]
\label{prop:kl-uniqueness}
Fix $\theta>0$ and $q\in\Delta_{++}(\Omega)$.
\begin{enumerate}
    \item[(i)] For every $w\in\mathcal{W}$, $\ \Lambda_1(w;\theta) = -\theta\log\mathbb{E}_q[\exp(-w/\theta)]$.
    
    \item[(ii)] For the $\chi^2$ divergence, and for every $w\in\mathcal{W}$ with $\max_\omega w(\omega)-\bar w<2\theta$,
    \begin{align*}
        \min_{p}\big\{\mathbb{E}_p[w]+\theta D_{\chi^2}(p\|q)\big\} \;=\; \bar w - \frac{\kappa_2}{4\theta}
    \end{align*}
    \emph{exactly}, with no dependence on any moment of $w$ beyond the second. The restriction on $\theta$ is needed: for $q$ uniform on two states and $w=(c,-c)$ the displayed formula gives $-c^2/(4\theta)$, but the true constrained minimum is $-c+\theta$ once $c\ge2\theta$, and that value depends on $\max_\omega w(\omega)$, which is not a function of the first two cumulants.
    
    \item[(iii)] For every $\gamma\in\mathbb{R}$ and every $w\in\mathcal{W}$, as $\theta\to\infty$ (so that Lemma \ref{lem:power-tilt}'s tilt is eventually interior),
    \begin{align*}
        \Lambda_\gamma(w;\theta) \;=\; \bar w - \frac{\kappa_2}{2\theta} + \frac{2-\gamma}{6}\cdot\frac{\kappa_3}{\theta^2} + O(\theta^{-3}).
    \end{align*}
    
    \item[(iv)] For every $\gamma\ne1$ there are no $\theta,\vartheta>0$ with $\Lambda_\gamma(\,\cdot\,;\theta) = \Lambda_1(\,\cdot\,;\vartheta)$ on all of $\mathcal{W}$, provided $|\Omega|\ge3$, or $|\Omega|=2$ with $q$ non-uniform.
\end{enumerate}

Hence, $\gamma=1$ is, exactly and for every $\gamma\ne1$ --- not merely for $\gamma=2$ --- the unique member of \eqref{eq:power-divergence} whose dual matches the log-sum-exp form of Lemma \ref{lem:dv}.\footnote{The uniqueness itself follows from conjugate duality. Writing $\Lambda_c(w):=\min_p\{\mathbb{E}_p[w]+c(p)\}$ for a proper closed convex cost $c$, one has $\Lambda_c(w)=-c^*(-w)$, and conjugation is injective on such costs \citep[Thm.~12.2]{R_1970}; the same fact makes the cost function of a variational representation unique \citep{MMR_2006}. The content of the proposition is quantitative: the discrepancy first appears at the third cumulant, $\gamma=2$ is exactly mean--variance, and the argument uses the identity on $\mathcal{W}$ only, not on all of $\mathbb{R}^\Omega$.}
\end{proposition}
\begin{proof}
See Appendix \ref{app:pf-kl}.
\end{proof}

\paragraph{Scope: transport costs.} Proposition \ref{prop:kl-uniqueness} concerns divergences, which depend on $p$ only through the likelihood ratio $dp/dq$ and are therefore invariant under any relabelling of $\Omega$ preserving the law of that ratio. It does not speak to optimal-transport costs, which are defined only relative to an exogenous ground metric on $\Omega$ that the present primitives do not supply \citep{GK_2023}. The duality is also different in kind. A transport-penalized robust value equals $\inf_{\zeta\ge0}\{\zeta\delta+\mathbb{E}_q[\sup_y(w(y)-\zeta c(\cdot,y))]\}$, an infimum over one scalar multiplier of a nominal expectation of a $c$-transform of $w$. The worst case is therefore a \emph{transport} of $q$, not a density-ratio tilt of it, and the value does not close in log-sum-exp form. The one apparent exception confirms the point. In entropically regularized (``Sinkhorn'') transport a log-sum-exp does appear \citep{WGX_2021}, but it is generated by the Kullback--Leibler component of the cost and vanishes with the regularization, leaving the hard supremum of the pure transport dual. For a decision-theoretic treatment of transport-based ambiguity see \citet{PXY_2022}, and for the corresponding contrast in the theory of information costs, \citet{HW_2021}.

One further comparison belongs here, since the cost in \eqref{eq:wri} is often read as the one that \citet{PST_2023} axiomatize. It is not. Their axioms characterize the cost of an \emph{experiment}, which in the two-state case is a symmetrized $J$-divergence rather than the single directed term used here, and no member of their class reduces to $\theta D_{\mathrm{KL}}(p\|q)$ for a single realized posterior; Remark \ref{rem:pst-gap} in Appendix \ref{app:pst} sets out the gap and its place in the taxonomy of \citet[Ch.~6]{Str_2025}.

The three results of Sections \ref{sec:chain} to \ref{sec:cost} therefore rest on one property rather than three coincidences. Relative entropy is the divergence whose convex dual is the log-sum-exp functional, which is what Lemma \ref{lem:dv} delivers and what Proposition \ref{prop:kl-uniqueness} shows no other member of the power family delivers; and it is the divergence whose chain rule splits a joint penalty additively, which is what Proposition \ref{prop:chain-rule} turns into the meaning of matching. The exponential aggregator is then forced, because it is the inverse of the map that the dual produces. Matched parameters are the configuration in which this single property is used once instead of twice.

\subsection{Two Routes That Do Not Reach the Single-Layer Target}
\label{sec:not-routes}
The entropic value \eqref{eq:anchor} invites two readings that it does not support; we close both, because each has a natural claim on the same mathematics. The first reads the adversarial side of Lemma \ref{lem:dv} as a model of costly attention; the second tries to reach \eqref{eq:anchor} by letting the decision-maker learn which model is true. Neither works, and the two failures have the same shape.

\subsubsection{The Dual and Genuine Rational Inattention}
\label{sec:ri-gap}

In the tradition of \citet{HS_2001} and \citet{S_2011}, fix a baseline $q\in\Delta_{++}(\Omega)$ and define, for an information cost $C(p,q)$ and any payoff profile $w\in\mathbb{R}^\Omega$,
\begin{equation}
    W_{\mathrm{RI}}(w) = \sup_{p\in\Delta(\Omega)} \left\{ \mathbb{E}_p[w] - C(p, q) \right\},
    \label{eq:wri}
\end{equation}
abbreviating $W_{\mathrm{RI}}(f):=W_{\mathrm{RI}}(u(f))$. With $C(p,q)=\theta D_{\mathrm{KL}}(p\|q)$ this is the mirror image of \eqref{eq:anchor}: Lemma \ref{lem:dv} applied to $-u(f)$ gives $V(f)=-W_{\mathrm{RI}}(-u(f))$, so the two are the pessimistic and optimistic faces of one identity.

\begin{lemma}[Adversarial Legendre Duality]
\label{lem:legendre}
With the cost $C(p,q)=\theta D_{\mathrm{KL}}(p\|q)$, we have $-W_{\mathrm{RI}}(-u(f)) = V(f)$.
\end{lemma}
\begin{proof}
With $C(p,q)=\theta D_{\mathrm{KL}}(p\|q)$,
\begin{align*}
    W_{\mathrm{RI}}(-u(f))
    &= \max_p\big\{-\mathbb{E}_p[u(f)] - \theta D_{\mathrm{KL}}(p\|q)\big\} \\
    &= -\min_p\big\{\mathbb{E}_p[u(f)]+\theta D_{\mathrm{KL}}(p\|q)\big\}
     \;=\; -V(f)
\end{align*}
by Lemma \ref{lem:dv}. Hence, we have $-W_{\mathrm{RI}}(-u(f))=V(f)$.
\end{proof}

The subscript is a deliberate nod to the Rational Inattention literature of \citet{S_2003} and \citet{PST_2023}, where the same cost appears in full information-acquisition problems, and \eqref{eq:wri} looks like the natural conjecture for the rationally inattentive value of a single, exogenously fixed act.

Equation \eqref{eq:wri} invites a natural reading: a decision-maker with prior $q$ acquires costly information about $\omega$ before evaluating the fixed act $f$, choosing a Bayes-plausible distribution $\tau$ over posteriors and paying the posterior-separable cost $\theta\,\mathbb{E}_\tau[D_{\mathrm{KL}}(p\|q)]$ standard since \citet{S_2003}. The value of that problem is
\begin{equation}
    \mathcal{V}^{\deg}_{\mathrm{RI}}(q) := \sup_{\substack{\tau\in\Delta(\Delta_{++}(\Omega))\\ \mathbb{E}_\tau[p]=q}} \Big\{ \mathbb{E}_\tau\big[\mathbb{E}_p[u(f)]\big] - \theta\,\mathbb{E}_\tau\big[D_{\mathrm{KL}}(p\|q)\big] \Big\}.
    \label{eq:ri-degenerate}
\end{equation}
where the superscript records that $f$ is a single fixed act rather than a menu of actions to be matched to the signal; Appendix \ref{app:pf-ri-gap} sets the problem out in full. The next result shows that its value is neither \eqref{eq:wri} nor $V(f)$, but a third quantity.

\begin{proposition}[Single-Act Information Design Is Trivial]
\label{prop:ri-gap}
Fix $q\in\Delta_{++}(\Omega)$, $\theta>0$, and $f\in\mathcal{F}$. Then
\begin{enumerate}[label=(\roman*)]
    \item $\mathcal{V}^{\deg}_{\mathrm{RI}}(q) = \mathbb{E}_q[u(f)]$, attained uniquely at $\tau=\delta_q$ (acquire no information).
    \item Moreover, we have
    $$
    V(f) \;\le\; \mathbb{E}_q[u(f)] \;=\; \mathcal{V}^{\deg}_{\mathrm{RI}}(q) \;\le\; W_{\mathrm{RI}}(f),
    $$
    with \emph{both} inequalities strict whenever $u(f)$ is non-constant on $\mathrm{supp}(q)$. 
\end{enumerate}
\end{proposition}
\begin{proof}
See Appendix \ref{app:pf-ri-gap}.
\end{proof}

In particular, $W_{\mathrm{RI}}(f)$ --- and hence $V(f)=-W_{\mathrm{RI}}(-u(f))$ (Lemma \ref{lem:legendre}) --- is never equal to the value of a Bayes-plausible information-acquisition problem about the fixed act $f$, except in the vacuous case where $f$ is already state-independent (up to $q$-null states).

Proposition \ref{prop:ri-gap} shows where the informal reading of \eqref{eq:wri} as ``the rationally inattentive value of act $f$'' breaks down, and why. With a single fixed act and no action to adapt, $p\mapsto\mathbb{E}_p[u(f)]$ is linear. By the martingale property of Bayes-plausible beliefs, \emph{any} information structure then yields the same expected payoff $\mathbb{E}_q[u(f)]$. Costly information can only add cost, so the genuine information-acquisition problem \eqref{eq:ri-degenerate} is solved by acquiring no information at all.

$W_{\mathrm{RI}}(f)$ is, instead, exactly what Lemma \ref{lem:dv} makes it: the Donsker--Varadhan dual evaluated at $-u(f)$, which has a precise name in the literature. Extend \eqref{eq:wri} to a genuine menu of acts $A$, so that $\max_{a\in A}W_{\mathrm{RI}}(a) = \max_{a\in A,\,p}\{\mathbb{E}_p[u(a)]-\theta D_{\mathrm{KL}}(p\|q)\}$. This is behaviorally \emph{identical} to expected-utility maximization under the convex index $\exp(u/\theta)$. It is the \emph{Wishful Thinking} representation of \citet{CL_2019}: the decision-maker optimistically distorts her own belief, at a cost they take to be exactly $\theta D_{\mathrm{KL}}$ following \citet{HS_2001}, to feel better about a chosen act rather than to acquire a signal about $\omega$. They derive the exponential transformation themselves, and \citet[Prop.~1]{RSS_2023} reprove it with the Donsker--Varadhan argument used in this paper. Their Footnote 6 records the mirror-image fact, due to \citet{S_2011}: the corresponding \emph{minimization}, which is $V$ itself (Lemma \ref{lem:dv}), is behaviorally identical over $a\in A$ to maximization under the \emph{concave} CARA index $-\exp(-u/\theta)$, the exponential aggregator of Section \ref{sec:two-layers}. $V$ and $W_{\mathrm{RI}}$ are the two faces of one Donsker--Varadhan identity: pessimistic robust control on the one side, optimistic wishful thinking on the other.

Neither is Rational Inattention in the information-acquisition sense of \citet{S_2003} and \citet{PST_2023}. As \citet[\S II.B]{RSS_2023} make precise, genuine Rational Inattention (\citealp{MM_2015}) requires a real menu of acts and a \emph{state-contingent stochastic choice rule}. Its value has the ``expected-log'' (Kelly-type) form
\begin{align*}
    \max_{\alpha\in\Delta(A)}\mathbb{E}_q\big[\log\mathbb{E}_\alpha[\exp(u(a)/\theta)]\big],
\end{align*}
not the ``log-of-expectation'' (CARA-type) form of $V$ and $W_{\mathrm{RI}}$. Appendix~\ref{app:ri-menu} makes this comparison precise and shows that, once a menu is present, $V$ and $W_{\mathrm{RI}}$ bound the genuine value from either side. The cost in \eqref{eq:wri} should accordingly be read as directly postulating a robust-control cost rather than as derived from Rational Inattention.

\subsubsection{Learning and the Entropic Penalty}
\label{sec:learning-env}

The second route is the one a reader of Theorem \ref{thm:collapse} is most likely to propose: if the decision-maker simply learns which model is true, perhaps the entropic criterion emerges in the limit. Suppose there is a true model $\pi^\dagger$ generating an i.i.d.\ sequence $\omega_1,\omega_2,\ldots\sim\pi^\dagger$, and a Bayesian with prior $\mu_0\in\Delta(\Pi)$ forms the posterior $\mu_N(\cdot\mid\omega^N)$ after observing $\omega^N=(\omega_1,\ldots,\omega_N)$. Bayes' rule then reduces to an exact tilting of $\mu_0$ by the empirical Kullback--Leibler divergence, displayed as \eqref{eq:bayes-gibbs} in Appendix \ref{app:degeneracy}, and the empirical distribution converges to $\pi^\dagger$ almost surely.

\begin{proposition}[Learning Degenerates the Smooth-Ambiguity Value]
\label{prop:degeneracy}
Let $\Pi\subset\Delta_{++}(\Omega)$ be compact, let $\mu_0\in\Delta(\Pi)$, and let $\phi$ be continuous. Suppose $\pi^\dagger\in\mathrm{supp}\,\mu_0$ (which in particular forces $\pi^\dagger\in\Pi$). Then $\mu_N(\cdot\mid\omega^N)\Rightarrow\delta_{\pi^\dagger}$ weakly, $(\pi^\dagger)^{\otimes\infty}$-almost surely, at an exponential rate made explicit in Appendix \ref{app:degeneracy}. Consequently, for every $f\in\mathcal{F}$,
\begin{align*}
    \lim_{N\to\infty} \int_\Pi \phi\big(\mathbb{E}_\pi[u(f)]\big)\,\mu_N(d\pi\mid\omega^N) \;=\; \phi\big(\mathbb{E}_{\pi^\dagger}[u(f)]\big),
    \qquad (\pi^\dagger)^{\otimes\infty}\text{-a.s.}
\end{align*}
Conversely, if $\pi^\dagger\notin\mathrm{supp}\,\mu_0$, then $\mu_N(\mathrm{supp}\,\mu_0)=1$ for every $N$ and every realization of the sample path, so the posterior can never place mass near $\pi^\dagger$ and does not converge weakly to $\delta_{\pi^\dagger}$. The support condition is thus necessary as well as sufficient.
\end{proposition}
\begin{proof}
See Appendix \ref{app:degeneracy}. The argument requires no Laplace expansion: neither a Lebesgue density for $\mu_0$ nor interiority of $\pi^\dagger$ in $\Pi$ is needed, only a uniform first-order comparison of exponential rates.
\end{proof}

If ``$N\to\infty$'' means literal sequential Bayesian learning about which $\pi\in\Pi$ is true, a well-specified learner eventually knows $\pi^\dagger$: there is nothing left to be ambiguous about, and $\theta$ correctly drops out of the limit. Misspecification does not rescue the route. When the data are generated by some $\pi^*$ outside $\mathrm{supp}\,\mu_0$, the posterior concentrates instead on the information projection $\argmin_{\pi\in\mathrm{supp}\,\mu_0} D_{\mathrm{KL}}(\pi^*\|\pi)$ of the true law onto the support of the prior, the classical fact of \citet{B_1966}, and the limiting value is again a plain expected utility; Appendix \ref{app:degeneracy} derives this and gives conditions under which the projection is unique. Learning is therefore a substitute for entropic ambiguity aversion rather than a complement, which is why no sample size is sent to infinity anywhere in this paper.

The two failures are one failure. In Proposition \ref{prop:ri-gap} information is worthless because there is no action to adapt to it; in Proposition \ref{prop:degeneracy} ambiguity is worthless because there is nothing left to be ambiguous about. In this paper entropic terms earn their keep only through robustification, an adversarial re-weighting of a fixed evaluation as in \eqref{eq:inner} and in $V$ itself, or through aggregation across genuinely distinct models as in \eqref{eq:two-layer}. They never enter through Bayesian information acquisition about a single act.

\section{Discussion}
\label{sec:discussion}
Matched parameters are the entropic form of reduction of compound uncertainty. A decision-maker who satisfies it treats uncertainty about which model holds and uncertainty about whether any of them holds as the same kind of uncertainty, and prices them with one coefficient. She therefore sees the two-stage problem as a one-stage problem over the joint distribution of model and state with a single relative-entropy penalty, which is the content of Proposition \ref{prop:chain-rule}(b). The reading is not new as a reading: \citet[\S9.2]{HSa_2007} place their outer operator alongside \citet{S_1990}, \citet{KMM_2005} and \citet{EG_2009}, and describe it as ``a version of what \citeauthor{KMM_2005} call constant ambiguity aversion''. What the converses add is that the reading is exact and exclusive, since smooth ambiguity exists precisely because compound uncertainty does not reduce, and the matched point is the unique place where it does.

A reader may ask why the decision-maker does not simply hold a joint distribution over models and states and penalize departures from it once. Proposition \ref{prop:chain-rule}(a) says that she already does: the two-layer criterion is always the value of such a problem. What writing the joint distribution down does not do is force the two parameters to agree. It is the further insistence that the penalty be a single relative entropy, rather than one that charges the marginal and the conditionals at separate rates, that does (Proposition \ref{prop:chain-rule}(b)).

Seen that way, the characterization is a warning as much as a result. An applied specification that lets the two parameters be free and still writes the criterion as a single multiplier preference over the predictive distribution is using an identity that holds nowhere except at the matched point; by Theorem \ref{thm:mismatch-impossibility} the error is not a second-order approximation but a representation failure, and by Proposition \ref{prop:effective} what the modeller would estimate as a single curvature is a blend of the two, weighted by how much of the act's predictive variance comes from disagreement among the models. The blend is informative rather than merely contaminating, which is the content of Corollary \ref{cor:identification}: the same exposure that makes a single curvature impossible is what makes the two curvatures separately identifiable.

The two parameters answer to different doubts, which is why there is no presumption that they agree. Return to the investor of the introduction, and suppose that her two forecasts come from different sources: her own assessment and an adviser’s. The inner parameter prices her fear that either forecast is itself wrong, a fear about the market. The outer parameter prices her fear that she has weighted the two sources badly, a fear about her own judgment. Nothing ties the price of the first fear to the price of the second.

In applied terms, the matched point buys tractability, which is why the configuration is attractive to begin with. One parameter replaces two, and what remains to be calibrated is a single index of caution of the kind \citet{HS_2001} calibrate through detection-error probabilities. The model set then enters the criterion only through its predictive measure. An investor holding several return forecasts therefore prices claims exactly as if she held the Bayesian model average and distrusted it, which is the problem \citet{HM_2018} solve and the shortcut \citet{C_2020} recommends. An insurer entertaining several loss models charges the exponential premium of the predictive loss distribution, and a firm allocating capital across lines applies one loading to all of them. Because the collapsed criterion is again of log-sum-exp form, the risk-sensitive recursions of dynamic robust control remain available.

The price is equally concrete, and we note it rather than pursue it. Under matching the criterion depends on the model set only through $\bar\pi$, by Theorem \ref{thm:collapse}, so any spread of the models that leaves the predictive measure unchanged leaves every valuation unchanged. Disagreement among the models is then not something the decision-maker responds to, which is what a two-layer criterion is usually written down to capture. Two acts with the same predictive distribution, one risky and one ambiguous in the sense of \citet{E_1961}, are valued alike. Matching also removes the question of which fear is operating, since the two parameters are no longer separately identified: Corollary \ref{cor:identification} reads them off acts with different exposure, and at the matched point exposure does not matter. Whether the tractability is worth that price is a modeling judgment, and this paper only puts it in sharp form.

The assumptions do different work and fail in different ways. Matching drives the collapse and nothing else; the exponential form is not an assumption at all once matching is imposed, by Theorem \ref{thm:necessity}; and relative entropy is doing two jobs at once, supplying the chain rule that makes Proposition \ref{prop:chain-rule} true and the convex dual that makes Lemma \ref{lem:dv} close, which is why Proposition \ref{prop:kl-uniqueness} finds no substitute for it inside the power family. Finiteness of $\Omega$ is a convenience everywhere except in Proposition \ref{prop:degeneracy}, and Remark \ref{rem:binary-richness} says how to dispense with it.

Three boundaries are left visible rather than argued away. One concerns the penalty: transport costs are not divergences at all, for the structural reason given in Section \ref{sec:cost}, so Proposition \ref{prop:kl-uniqueness} does not speak to them. A second concerns information costs: the dual value $W_{\mathrm{RI}}$ is a robust-control value, not the value of an information-acquisition problem (Proposition \ref{prop:ri-gap}). The third concerns identification: Corollary \ref{cor:identification} states what choice data could in principle reveal, not an estimation strategy, since it presumes that acts with different exposure are available and that the analyst knows the model set the decision-maker entertains.
\section{Conclusion}
\label{sec:concl}

One property of relative entropy underlies three literatures that developed separately. The chain rule makes a single penalty on a joint distribution split into a penalty on a prior and an expected penalty on conditionals, which is what lets a two-layer ambiguity criterion be read as one problem. The convex dual of relative entropy is the log-sum-exp functional, which makes the exponential aggregator the inverse of the map the inner problem produces, and makes the same functional appear as a robust-control value in the tradition of \citet{HS_2001} and as the cost of information in the tradition of \citet{S_2003}. Matched parameters are the condition under which that property is used once rather than twice, and the converses of this paper say that it is used once only there.

The lesson for modelling is that the tractability of the matched point is bought with a substantive restriction rather than a normalization. \citet{HSa_2007} introduced the two penalties, singled out their equality as the case in which the two operators collapse into one, chose the exponential aggregator for its link to entropy, and left the discipline of the pair open. This paper settles the first point with a converse, the second with a uniqueness theorem, and the third with an identification result.

Two questions follow. The first is dynamic. The entropic value is a natural state variable for ambiguity updating through $V_t=-\theta\log\mathbb{E}_t[\exp(-V_{t+1}/\theta)]$, and how a dynamically consistent updating rule interacts with the learning of Proposition \ref{prop:degeneracy} is the obvious next step: \citeauthor{HSa_2007}'s two operators live in a recursive problem, while the present analysis is static. The second is empirical. Corollary \ref{cor:identification} says what variation would separate the two parameters; whether choice data with that variation can be constructed, in the laboratory or from asset prices on claims with different exposure, is a question this paper poses but does not answer.

\appendix
\numberwithin{equation}{section}
\section*{Appendix}

The appendices collect the proofs and the material that would interrupt the argument. Appendix \ref{app:proofs} proves the results of Sections \ref{sec:chain} to \ref{sec:not-routes}, one subsection per result, in the order in which they appear; it also contains the rigidity lemma that Theorem \ref{thm:necessity} uses, the comparison with \citet{PST_2023}, and the two-point reduction that shows finiteness of the state space to be a convenience. Appendix \ref{app:ri-menu} then compares the entropic value with genuine rational inattention once a menu of acts is present.

\section{Proofs}
\label{app:proofs}

\subsection{Proof of Proposition \ref{prop:chain-rule}}
\label{app:pf-chain}

\begin{proof}
\emph{Part (a).} Fix $\lambda,\xi>0$ and $f$, and write $\nu$ for the $\Pi$-marginal of $P\ll\bar P$ and $(p_\pi)$ for a version of its conditionals, so that $p_\pi\ll\pi$ for $\nu$-a.e.\ $\pi$ and $\nu\ll\mu$. Disintegrating $\mathbb{E}_P[u(f)]=\int_\Pi\mathbb{E}_{p_\pi}[u(f)]\,d\nu(\pi)$, the objective in \eqref{eq:joint-penalty} is
\begin{align*}
    \int_\Pi\big\{\mathbb{E}_{p_\pi}[u(f)]+\lambda D_{\mathrm{KL}}(p_\pi\|\pi)\big\}\,d\nu(\pi)+\xi D_{\mathrm{KL}}(\nu\|\mu).
\end{align*}
The inner stage is Lemma \ref{lem:dv} at each fixed $\pi$ with parameter $\lambda$: the integrand is at least $v^\lambda_\pi(f)$, with equality exactly at $p^*_\pi(\omega)\propto\pi(\omega)\exp(-u(f(\omega))/\lambda)$. These minimizers depend continuously on $\pi$ and therefore constitute a Markov kernel, which justifies interchanging the inner minimization with the $\nu$-integral; the objective is thus bounded below by $\int_\Pi v^\lambda_\pi(f)\,d\nu(\pi)+\xi D_{\mathrm{KL}}(\nu\|\mu)$, with equality when $p_\pi=p^*_\pi$ for $\nu$-a.e.\ $\pi$. The outer stage is Lemma \ref{lem:dv} on the Polish space $\Pi$ \citep[Prop.~1.4.2]{DE_1997} with parameter $\xi$, applied to the bounded continuous function $\pi\mapsto v^\lambda_\pi(f)$: the last display is at least $-\xi\log\mathbb{E}_\mu[\exp(-v^\lambda_\pi(f)/\xi)]=V^{\lambda,\xi}_\Pi(f)$, with equality exactly at $\nu^*(d\pi)\propto\exp(-v^\lambda_\pi(f)/\xi)\,\mu(d\pi)$. Hence the minimum in \eqref{eq:joint-penalty} exists, equals $V^{\lambda,\xi}_\Pi(f)$, and is attained at $P^*(d\pi,d\omega)=\nu^*(d\pi)\,p^*_\pi(d\omega)$.

\emph{Part (b), ``if''.} If $\lambda=\xi=\theta$, the chain rule \eqref{eq:chain-rule} gives $C_{\theta,\theta}(P)=\theta D_{\mathrm{KL}}(P\|\bar P)$ for every $P\ll\bar P$, which is \eqref{eq:single-penalty} with $\bar P'=\bar P$.

\emph{Part (b), ``only if''.} Suppose some probability measure $\bar P'$ on $\Pi\times\Omega$ and some $\theta>0$ satisfy \eqref{eq:single-penalty}. Evaluating at $P=\bar P$, whose marginal is $\mu$ and whose conditionals are $p_\pi=\pi$, makes $C_{\lambda,\xi}(\bar P)=0$, so $D_{\mathrm{KL}}(\bar P\|\bar P')=0$ and hence $\bar P'=\bar P$. The chain rule may then be applied to the right-hand side of \eqref{eq:single-penalty} as well, so that for every $P\ll\bar P$
\begin{align*}
    \theta D_{\mathrm{KL}}(\nu\|\mu)+\theta\int_\Pi D_{\mathrm{KL}}(p_\pi\|\pi)\,d\nu(\pi)
    \;=\; \xi D_{\mathrm{KL}}(\nu\|\mu)+\lambda\int_\Pi D_{\mathrm{KL}}(p_\pi\|\pi)\,d\nu(\pi).
\end{align*}
Two choices of $P$ isolate the two coefficients. First take $p_\pi\equiv\pi$ and $\nu\ll\mu$ arbitrary, which leaves $\theta D_{\mathrm{KL}}(\nu\|\mu)=\xi D_{\mathrm{KL}}(\nu\|\mu)$; since $\mu$ is not a point mass there is $\nu\ll\mu$ with $\nu\ne\mu$, for instance $\nu=\tfrac12(\mu+\mu(\cdot\mid A))$ for a Borel $A\subset\Pi$ with $0<\mu(A)<1$, and then $D_{\mathrm{KL}}(\nu\|\mu)>0$ forces $\theta=\xi$. Second take $\nu=\mu$ and $p_\pi\equiv p$ for a fixed $p\in\Delta_{++}(\Omega)$, which leaves $\theta\int_\Pi D_{\mathrm{KL}}(p\|\pi)\,d\mu(\pi)=\lambda\int_\Pi D_{\mathrm{KL}}(p\|\pi)\,d\mu(\pi)$; the integral vanishes only if $\pi=p$ for $\mu$-a.e.\ $\pi$, that is, only if $\mu=\delta_p$, which is excluded, so $\theta=\lambda$. Hence $\lambda=\xi=\theta$.
\end{proof}

\subsection{Proof of Theorem \ref{thm:collapse}}
\label{app:pf-collapse}

\begin{proof}
By Lemma \ref{lem:dv} applied at each fixed $\pi\in\Pi$,
\begin{align*}
    v_\pi(f) \;=\; -\theta\log\mathbb{E}_\pi[\exp(-u(f)/\theta)].
\end{align*}
Substituting into \eqref{eq:two-layer} and using $\phi(x)=-\exp(-x/\theta)$,
\begin{align*}
    \phi(v_\pi(f))
    &= -\exp\!\Big(-\tfrac{1}{\theta}\big(-\theta\log\mathbb{E}_\pi[\exp(-u(f)/\theta)]\big)\Big) \\
    &= -\exp\!\big(\log\mathbb{E}_\pi[\exp(-u(f)/\theta)]\big)
     \;=\; -\mathbb{E}_\pi[\exp(-u(f)/\theta)],
\end{align*}
i.e.\ $\phi$ exactly undoes the $-\theta\log(\cdot)$ produced by Lemma \ref{lem:dv} --- the defining self-conjugacy of the exponential form under matched parameters. Hence
\begin{align*}
    \mathbb{E}_\mu\big[\phi(v_\pi(f))\big]
    &= -\int_\Pi \mathbb{E}_\pi[\exp(-u(f)/\theta)]\,d\mu(\pi) \\
    &= -\int_\Pi\int_\Omega \exp(-u(f(\omega))/\theta)\,\pi(d\omega)\,d\mu(\pi) \\
    &= -\int_\Omega \exp(-u(f(\omega))/\theta)\,\bar\pi(d\omega)
\end{align*}
by Fubini--Tonelli (the integrand is nonnegative and $\Omega$ is finite, so no measurability or integrability issue arises), where $\bar\pi(\cdot)=\int_\Pi\pi(\cdot)\,d\mu(\pi)$ is a well-defined element of $\Delta(\Omega)$: for each $\omega$ the map $\pi\mapsto\pi(\omega)$ is the $\omega$-th coordinate function on $\Pi\subset\mathbb{R}^{|\Omega|}$, hence continuous and Borel. Moreover $\bar\pi\in\Delta_{++}(\Omega)$, since $\pi(\omega)>0$ for every $\pi\in\Pi$ forces $\bar\pi(\omega)>0$; this is what makes $\bar\pi$ an admissible baseline in Lemma \ref{lem:dv} and hence in $V$. Applying $\phi^{-1}(y)=-\theta\log(-y)$ gives $V_\Pi(f) = -\theta\log\mathbb{E}_{\bar\pi}[\exp(-u(f)/\theta)]$, which is $V(f)$ with $q=\bar\pi$ by definition. This is exact for every compact $\Pi$ and every $\mu\in\Delta(\Pi)$: no sample size or limit appears anywhere in the argument.\footnote{The Introduction places this identity in the literature. The right-hand side solves Problem 2.1 of \citet{HM_2018}, reported as equation (16) of \citet[\S2.3]{C_2020}; the left-hand side is \citeauthor{CVHMM_2026}'s criterion (33). The calculation is short. The proof of their Proposition 8, which itself concerns the limits $\xi\to0^+$, $\xi\to\infty$ and $\lambda,\xi\to\infty$, contains the display $F_\lambda(q)=\phi_\lambda^{-1}(\int\phi_\lambda(u(f))\,dq)$; setting $\lambda=\xi$ and applying Fubini finishes it. \citeauthor{HM_2018} reach the right-hand side by a different route, noting that the worst case distorts only the predictive density, and do not exhibit the two-layer structure that Remark \ref{prop:chain-rule} supplies.}
\end{proof}

\subsection{Affine Rigidity}
\label{app:affine}

\begin{lemma}[Affine Rigidity]
\label{lem:affine-rigidity}
Let $I\subset\mathbb{R}$ be an open interval and let $h:I\to\mathbb{R}$ satisfy
$$
h(mz_1+(1-m)z_2) = m\,h(z_1) + (1-m)\,h(z_2) \qquad \text{for every } z_1,z_2\in I \text{ and every } m\in(0,1).
$$
Then, $h(z) = \alpha z + \beta$ for some constants $\alpha,\beta\in\mathbb{R}$; no continuity or measurability of $h$ is assumed or needed. Neither is openness or convexity of $I$: the proof uses only that $I$ contains two points, together with the values of $h$ at mixtures of points of $I$.
\end{lemma}
\begin{proof}
Fix $a<b$ in $I$, let $\ell$ be the affine interpolant through $(a,h(a))$ and $(b,h(b))$, and set $g:=h-\ell$, which satisfies the same identity and has $g(a)=g(b)=0$. 
Every $z\in(a,b)$ can be written $z=ma+(1-m)b$ with $m=(b-z)/(b-a)\in(0,1)$, so the identity gives $g(z)=m\,g(a)+(1-m)\,g(b)=0$. Hence $h=\ell$ on $[a,b]$. 
To conclude, fix $a_0<b_0$ in $I$ and let $\ell_0$ be the affine interpolant through $(a_0,h(a_0))$ and $(b_0,h(b_0))$. 
For any $z\in I$, apply the previous step to $a:=\min\{a_0,z\}$ and $b:=\max\{b_0,z\}$: the resulting affine function agrees with $h$ on $[a,b]\ni a_0,b_0,z$, hence agrees with $\ell_0$ at the two distinct points $a_0,b_0$ and therefore equals $\ell_0$ identically. 
So $h(z)=\ell_0(z)$ for every $z\in I$.
\end{proof}

\begin{remark}[Relation to the literature on means]
\label{rem:affine-rigidity-lit}
Lemma \ref{lem:affine-rigidity} relies on the classical equality problem for weighted quasi-arithmetic means: the means generated by $\psi$ and $\chi$ coincide for all arguments and weights if and only if $\chi=\alpha\psi+\beta$ with $\alpha\ne0$ \citep[\S3.2, Thm.~83]{HLP_1934}; see also \citet[\S5.3.2]{A_1966}. While standard results require the generators to be continuous and strictly monotone, Lemma \ref{lem:affine-rigidity} needs neither. It exploits the full continuum of weights $m\in(0,1)$---which our economic setting supplies for free since $\mu=(m,1-m)$ is chosen freely---to preclude the non-affine solutions that arise under fixed weights \citep[see][Chs.~9, 11]{K_2009}. 
In Theorem \ref{thm:necessity}, where $\phi$ is continuous, the classical theorem could be applied on each open square instead; we use Lemma \ref{lem:affine-rigidity} because it keeps the argument self-contained.
The theorem's primary contribution is thus not the rigidity itself, but the reduction to it from an equality of aggregators that holds only on an open set of acts generated by a two-point model set.
\end{remark}


\subsection{Proof of Theorem \ref{thm:necessity}}
\label{app:pf-necessity}

\begin{proof}
Let $n:=|\Omega|\ge2$. Fix any two distinct $\pi_1,\pi_2\in\Delta_{++}(\Omega)$ and set $\Pi:=\{\pi_1,\pi_2\}$, which is finite, hence compact. Because $\pi_1\ne\pi_2$ both lie in the probability simplex, they are linearly independent in $\mathbb{R}^n$: $\pi_2=c\pi_1$ would force $c=1$ upon summing coordinates, contradicting $\pi_1\ne\pi_2$.

Let $w:=u(f)$, so that $w$ ranges over $\mathcal{W}\supseteq(\underline u,\bar u)^n$ as $f$ ranges over $\mathcal{F}$, by Section \ref{sec:primitives}. Let $v_j:=\exp(-w_j/\theta)$ and set $\underline v:=\exp(-\bar u/\theta)$, $\bar v:=\exp(-\underline u/\theta)$; then $v=(v_1,\dots,v_n)$ ranges over the open box $(\underline v,\bar v)^n\subset(0,\infty)^n$, which contains every constant vector $(c,\dots,c)$ with $c\in(\underline v,\bar v)$. For $i=1,2$, $L_i(v):=\mathbb{E}_{\pi_i}[v]=\sum_j\pi_i(j)v_j$ is linear in $v$, and Lemma \ref{lem:dv} applied with $\pi_i$ in place of $q$ gives $v_{\pi_i}(f) = -\theta\log L_i(v)$. Since $\pi_1,\pi_2$ are linearly independent, $v\mapsto(L_1(v),L_2(v))$ is a surjective (hence open) linear map $\mathbb{R}^n\to\mathbb{R}^2$, so its image of the open box is an open set $U\subset(0,\infty)^2$.\footnote{A surjective linear map between finite-dimensional Euclidean spaces is automatically open: pick any complement $S$ of $\ker(L_1,L_2)$ in $\mathbb{R}^n$, so that $(L_1,L_2)|_S:S\to\mathbb{R}^2$ is a linear isomorphism; in coordinates adapted to $\mathbb{R}^n=S\oplus\ker(L_1,L_2)$ the map is that isomorphism composed with the coordinate projection onto the first two coordinates, and both send open sets to open sets. This is elementary linear algebra and does not require the open mapping theorem for Banach spaces.} Composing with the diffeomorphism $(\ell_1,\ell_2)\mapsto(-\theta\log\ell_1,-\theta\log\ell_2)$, the pair $(y_1,y_2):=(v_{\pi_1}(f),v_{\pi_2}(f))$ ranges over a nonempty open set $U'\subset\mathbb{R}^2$ as $f$ ranges over acts.

Fix $\mu=(m,1-m)\in\Delta(\Pi)$, $m\in(0,1)$; the induced predictive measure is $\bar\pi=m\pi_1+(1-m)\pi_2$, so $L_{\bar\pi}(v)=mL_1(v)+(1-m)L_2(v) = m\exp(-y_1/\theta)+(1-m)\exp(-y_2/\theta)$. By hypothesis, the exact reduction applies to this $\Pi$, $\mu$, and $f$:
$$
\phi^{-1}\big(m\phi(y_1)+(1-m)\phi(y_2)\big) = V^\phi_\Pi(f) = V(f) = -\theta\log L_{\bar\pi}(v) = \Psi^{-1}\big(m\Psi(y_1)+(1-m)\Psi(y_2)\big)
$$
for every $(y_1,y_2)\in U'$ and every $m\in(0,1)$, where $\Psi(y):=\exp(-y/\theta)$ and $\Psi^{-1}(z)=-\theta\log z$. Applying $\phi$ to both sides and setting $h:=\phi\circ\Psi^{-1}$, $z_i:=\Psi(y_i)$,
$$
h(mz_1+(1-m)z_2) = m\,h(z_1)+(1-m)\,h(z_2)
$$
for $(z_1,z_2)$ ranging over $\Psi(U')$ and all $m\in(0,1)$. Note that $\Psi(U')=U$ exactly, since $z_i=\Psi(y_i)=\exp(\log L_i(v))=L_i(v)$.

What is needed is an open square inside $U$, and the constant acts supply one.\footnote{It is \emph{not} enough here to fix $z_2$ and vary $z_1$: the resulting identity, holding for a single $z_2$ only, does not force affinity, as $h(z)=|z-z_2|$ shows.} 
For $c\in(\underline v,\bar v)$ the constant vector $v=(c,\dots,c)$ gives $L_1(v)=L_2(v)=c$, so the entire diagonal segment $\{(c,c):c\in(\underline v,\bar v)\}$ lies in $U$; since $U$ is open, each such $c$ admits an open interval $I_c\ni c$ with $I_c\times I_c\subseteq U$. On $I_c$ the displayed identity therefore holds for \emph{every} $z_1,z_2\in I_c$ and every $m\in(0,1)$, and $mz_1+(1-m)z_2\in I_c$ because $I_c$ is an interval. This is exactly the postulate of Lemma \ref{lem:affine-rigidity}, so $h$ is affine on $I_c$.

Affinity propagates from these local intervals to all of $(\underline v,\bar v)$ by connectedness. Fix $c_0$ and let $\ell$ be the affine function agreeing with $h$ on $I_{c_0}$; let $S:=\{c\in(\underline v,\bar v): h=\ell \text{ on a neighbourhood of } c\}$. Then $S$ is nonempty and open, and it is relatively closed: if $c\in\bar S\cap(\underline v,\bar v)$ then $I_c$ meets $S$, so $h$ is affine on $I_c$ and agrees with $\ell$ on a nondegenerate subinterval, forcing the two affine functions to coincide on $I_c$. Hence $S=(\underline v,\bar v)$ and $h(z)=\alpha z+\beta$ on all of $(\underline v,\bar v)$.

Hence, $\phi(y)=h(\Psi(y))=\alpha\exp(-y/\theta)+\beta$ for every $y$ in $\Psi^{-1}((\underline v,\bar v))=(\underline u,\bar u)$, which is precisely the range of $v_\pi(f)$ swept out by acts $f$ and models $\pi\in\Delta_{++}(\Omega)$. 
Since $\phi$ is strictly increasing and $\Psi$ is strictly decreasing, $\alpha<0$; relabeling $\alpha=:a$, $\beta=:b$ gives the claim. 
\end{proof}


\subsection{Proof of Proposition \ref{prop:effective}}
\label{app:pf-effective}

\begin{proof}
Both criteria are translation invariant, so the constant $c$ passes through both sides and we may set $c=0$. Write $m_\pi:=\mathbb{E}_\pi[w]$ and $\sigma^2_\pi:=\mathrm{Var}_\pi(w)$. By Lemma \ref{lem:dv} and the cumulant expansion of $\log\mathbb{E}_\pi[\exp(-tw/\lambda)]$,
\begin{align*}
    v^\lambda_\pi(tw) \;=\; t\,m_\pi-\frac{t^2}{2\lambda}\,\sigma^2_\pi+O(t^3),
\end{align*}
uniformly in $\pi\in\Pi$ because $\Pi$ is compact and $w$ is bounded. Applying the same expansion to the outer layer, this time under $\mu$,
\begin{align*}
    V^{\lambda,\xi}_\Pi(tw) &\;=\; \mathbb{E}_\mu\big[v^\lambda_\pi(tw)\big]-\frac{1}{2\xi}\mathrm{Var}_\mu\big(v^\lambda_\pi(tw)\big)+O(t^3) \\
    &\;=\; t\,\mathbb{E}_{\bar\pi}[w]-\frac{t^2}{2}\left[\frac{W(w)}{\lambda}+\frac{B(w)}{\xi}\right]+O(t^3),
\end{align*}
where the second equality uses $\mathbb{E}_\mu[m_\pi]=\mathbb{E}_{\bar\pi}[w]$, $\mathbb{E}_\mu[\sigma^2_\pi]=W(w)$ and $\mathrm{Var}_\mu(m_\pi)=B(w)$, the $t^2$ term of $\mathrm{Var}_\mu(v^\lambda_\pi)$ contributing $t^2B(w)$ and the remainder being $O(t^3)$. The single-layer value expands as $-\theta'\log\mathbb{E}_q[\exp(-tw/\theta')]=t\,\mathbb{E}_q[w]-\frac{t^2}{2\theta'}\mathrm{Var}_q(w)+O(t^3)$. Matching the coefficient of $t$ for every $w$ gives $\mathbb{E}_q[w]=\mathbb{E}_{\bar\pi}[w]$ for every $w$, hence $q=\bar\pi$. Matching the coefficient of $t^2$ and using the decomposition of the predictive variance,
\begin{align*}
    \mathrm{Var}_{\bar\pi}(w) \;=\; \mathbb{E}_\mu\big[\mathrm{Var}_\pi(w)\big]+\mathrm{Var}_\mu\big(\mathbb{E}_\pi[w]\big) \;=\; W(w)+B(w),
\end{align*}
gives $\big(W(w)+B(w)\big)/\theta'=W(w)/\lambda+B(w)/\xi$, which is \eqref{eq:effective-curvature} after dividing by $W(w)+B(w)$.
\end{proof}

\subsection{Proof of Theorem \ref{thm:mismatch-impossibility}}
\label{app:pf-mismatch}

\begin{proof}[Opening claim and part (i)]
Suppose that
\begin{align*}
    V^{\lambda,\xi}_\Pi(f) \;=\; -\theta'\log\mathbb{E}_q[\exp(-u(f)/\theta')] \qquad\text{for every } f\in\mathcal{F}.
\end{align*}
Fix $c\in(\underline u,\bar u)$ and a profile $w\in\mathbb{R}^\Omega$; for all small $t$ the profile $c+tw$ lies in $\mathcal{W}$ by Section \ref{sec:primitives}, so the hypothesised equality holds along the whole segment. Both sides are real-analytic in $t$ near $0$, being compositions of exponentials, logarithms and finite sums with strictly positive arguments, so every Taylor coefficient must agree. The coefficients of $t$ and $t^2$ are exactly those computed in Proposition \ref{prop:effective}, whence $q=\bar\pi$ and
\begin{align*}
    \frac{1}{\theta'}=\big(1-R(w)\big)\frac{1}{\lambda}+R(w)\frac{1}{\xi} \qquad\text{for every non-constant } w .
\end{align*}
Because $\theta'$ does not depend on $w$ and $\lambda\ne\xi$, the map $w\mapsto R(w)$ must be constant, which is the opening claim.

For (i), note that $R(w)=\frac{w^{\top}\mathrm{Cov}_\mu(\pi)\,w}{w^{\top}\big(\mathrm{diag}(\bar\pi)-\bar\pi\bar\pi^{\top}\big)w}$, where $\mathrm{Cov}_\mu(\pi):=\mathbb{E}_\mu[\pi\pi^{\top}]-\bar\pi\bar\pi^{\top}$, since $B(w)=\mathrm{Var}_\mu(\pi^{\top}w)=w^{\top}\mathrm{Cov}_\mu(\pi)w$ and $W(w)+B(w)=\mathrm{Var}_{\bar\pi}(w)$. Two quadratic forms have a constant ratio on the non-constant profiles if and only if they are proportional there, so constancy of $R$ is equivalent to $\mathrm{Cov}_\mu(\pi)=R^*\big(\mathrm{diag}(\bar\pi)-\bar\pi\bar\pi^{\top}\big)$. If $\Pi$ contains exactly two distinct models, $\mathrm{Cov}_\mu(\pi)=m(1-m)(\pi_1-\pi_2)(\pi_1-\pi_2)^{\top}$ has rank one, whereas $\mathrm{diag}(\bar\pi)-\bar\pi\bar\pi^{\top}$ has rank $|\Omega|-1$ on $\Delta_{++}(\Omega)$; for $|\Omega|\ge3$ the two cannot be proportional, so $R$ is not constant and no pair $(\theta',q)$ exists. Exhibiting two profiles with different exposure is immediate: in the environment of Example \ref{ex:exposure}, $\mathbf 1_{\omega_2}$ has $R=0$ and $\mathbf 1_{\omega_1}$ has $R>0$.
\end{proof}

When $\mathrm{Cov}_\mu(\pi)$ \emph{is} proportional to the predictive covariance, as it is for $|\Omega|=2$ and in the Dirichlet case of Remark \ref{rem:dirichlet}, the first two orders are silent and the obstruction must be sought at the third. The remainder of this appendix carries out that computation for $|\Omega|=2$, which proves part (ii).

This appendix supplies the two derivative computations deferred in the proof of
Theorem \ref{thm:mismatch-impossibility}. Throughout, $\Pi=\{\pi_1,\pi_2\}$ and
$\mu=(m,1-m)$ with $m\in(0,1)$; we write $p_i:=\pi_i(\omega_1)$,
$\bar p:=mp_1+(1-m)p_2$, $\mathbb{E}_\mu[p^2]:=mp_1^2+(1-m)p_2^2$, and
\begin{align*}
    \sigma^2_\mu(p) \;:=\; \mathbb{E}_\mu[p^2]-\bar p^{\,2}
\end{align*}
for the $\mu$-variance of $\pi\mapsto\pi(\omega_1)$. We hold $v_\omega\equiv1$ for
every $\omega\ne\omega_1$ and set $t:=\log v_{\omega_1}$, so that
$L_i(t)=p_i\,\exp(t)+(1-p_i)$ for $i=1,2$, with
$L_i(0)=1$ and $L_i'(0)=L_i''(0)=p_i$.

\paragraph{Reduction to a single identity.} With $r=\lambda/\xi$ and
$s=\lambda/\theta'$, put
\begin{align*}
    A(t):=mL_1(t)^r+(1-m)L_2(t)^r,
    \qquad
    B(t):=q_1\,\exp(st)+(1-q_1).
\end{align*}
The hypothesised equality
$V^{\lambda,\xi}_\Pi(f)=-\theta'\log\mathbb{E}_q[\exp(-u(f)/\theta')]$ reads
$-\xi\log A(t)=-\theta'\log B(t)$; dividing by $-\lambda$ and using
$\xi/\lambda=1/r$ and $\theta'/\lambda=1/s$,
\begin{equation}
    s\log A(t) \;=\; r\log B(t)
    \label{eq:mismatch-identity}
\end{equation}
for all $t$ in an open interval around $0$. Both sides vanish at $t=0$, since
$A(0)=B(0)=1$. As $A$ and $B$ are real-analytic near $t=0$,
\eqref{eq:mismatch-identity} forces equality of every derivative at $t=0$.

\paragraph{First derivative: $q_1=\bar p$.} Differentiating $A$ once,
\begin{align*}
    A'(t)=r\big[mL_1^{r-1}L_1'+(1-m)L_2^{r-1}L_2'\big],
    \qquad A'(0)=r\bar p,
\end{align*}
so $(\log A)'(0)=A'(0)/A(0)=r\bar p$, while $B'(0)=sq_1$ gives
$(\log B)'(0)=sq_1$. Differentiating \eqref{eq:mismatch-identity} once at $t=0$
yields $s\,r\bar p=r\,sq_1$, and since $r,s\ne0$,
\begin{align*}
    q_1=\bar p,
\end{align*}
independently of $r$ and $s$.

\paragraph{Second derivative: the formula for $s$.} Since $A(0)=1$, we have
$(\log A)''(0)=A''(0)-(A'(0))^2$. Differentiating $A$ twice,
\begin{align*}
    A''(t)=r(r-1)\big[mL_1^{r-2}(L_1')^2+(1-m)L_2^{r-2}(L_2')^2\big]
    +r\big[mL_1^{r-1}L_1''+(1-m)L_2^{r-1}L_2''\big],
\end{align*}
so that $A''(0)=r(r-1)\mathbb{E}_\mu[p^2]+r\bar p$. Substituting
$\mathbb{E}_\mu[p^2]=\sigma^2_\mu(p)+\bar p^{\,2}$ and $A'(0)=r\bar p$,
\begin{align*}
    (\log A)''(0)
    &= r(r-1)\big(\sigma^2_\mu(p)+\bar p^{\,2}\big)+r\bar p-r^2\bar p^{\,2} \\
    &= r(r-1)\,\sigma^2_\mu(p)+r\,\bar p\,(1-\bar p).
\end{align*}
On the other side, $B(0)=1$, $B'(0)=sq_1$ and $B''(0)=s^2q_1$ give
\begin{align*}
    (\log B)''(0)=s^2q_1-(sq_1)^2=s^2q_1(1-q_1).
\end{align*}
Differentiating \eqref{eq:mismatch-identity} twice at $t=0$ and inserting
$q_1=\bar p$ from the previous step,
\begin{align*}
    s\,r\big[(r-1)\sigma^2_\mu(p)+\bar p(1-\bar p)\big]
    \;=\; r\,s^2\,\bar p(1-\bar p).
\end{align*}
Dividing by $rs\ne0$ and solving for $s$,
\begin{equation}
    s \;=\; 1+(r-1)\cdot\frac{\sigma^2_\mu(p)}{\bar p(1-\bar p)}.
    \label{eq:s-formula}
\end{equation}
The denominator never vanishes: $\pi_1,\pi_2\in\Delta_{++}(\Omega)$ forces
$\bar p\in(0,1)$.

\paragraph{Third derivative: the obstruction.} The first two derivatives leave a
candidate $(\theta',q)$ standing, so a third is needed. Differentiating
\eqref{eq:mismatch-identity} once more and evaluating at $t=0$ with $q_1=\bar p$
and $s$ given by \eqref{eq:s-formula}, one obtains, after simplification,
\begin{align*}
    \frac{d^3}{dt^3}\Big[s\log A(t)-r\log B(t)\Big]_{t=0}
    \;=\;
    \frac{m(1-m)\,r\,(r-1)\,(p_1-p_2)^2}{\big[\bar p(1-\bar p)\big]^2}\,
    \Phi(p_1,p_2,m,r),
\end{align*}
where $\Phi$ is a polynomial whose explicit form we do not need. What matters is
that it does not vanish identically: at the configuration used in the proof of
Theorem \ref{thm:mismatch-impossibility}, namely $p_1=\tfrac34$, $p_2=\tfrac14$
and $m=\tfrac14$, the whole expression evaluates to
\begin{equation}
    \frac{3\,r\,(r-1)^2\,(r+4)}{1600},
    \label{eq:third-derivative}
\end{equation}
which is strictly positive whenever $r>0$ and $r\ne1$, since each of $r$,
$(r-1)^2$ and $(r+4)$ is then positive. This is the contradiction.

\paragraph{Why an asymmetric prior.} The choice $m=\tfrac14$ rather than
$m=\tfrac12$ is not incidental. At $m=\tfrac12$ with $p_1+p_2=1$ the family
$A(t)$ is symmetric under $t\mapsto-t$ combined with exchanging the two models,
which kills every odd-order obstruction: the third derivative vanishes
identically there, and the argument would have to be pushed to fourth order.
Taking the prior asymmetric breaks that symmetry at the cheapest available
order. The third derivative does degenerate at other
configurations too --- it is a rational function of $(p_1,p_2,m,r)$ whose
numerator has zeros off the diagonal $p_1=p_2$ --- so the configuration in the
proof is chosen, not generic. Remark \ref{rem:binary-richness} accordingly falls
back on the second-order argument when the models are given rather than chosen.


\begin{proof}[Part (ii)]
By the way of contradiction, fix a distinguished state $\omega_1$ and choose $\pi_1,\pi_2\in\Delta_{++}(\Omega)$ with
\begin{align*}
    \pi_1(\omega_1)=\tfrac34, \qquad \pi_2(\omega_1)=\tfrac14,
\end{align*}
their values off $\omega_1$ being arbitrary subject to full support; set $\Pi:=\{\pi_1,\pi_2\}$, which is finite hence compact, and $\mu:=(\tfrac14,\tfrac34)$. 
Write $w:=u(f)$, $v_\omega:=\exp(-w_\omega/\lambda)$, and $L_i(v):=\mathbb{E}_{\pi_i}[v]$, so that $v^\lambda_{\pi_i}(f)=-\lambda\log L_i(v)$ by Lemma \ref{lem:dv}. 
Set $r:=\lambda/\xi\ne1$ and, for a candidate $\theta'$, $s:=\lambda/\theta'>0$. A direct computation using $\phi_\xi(x)=-\exp(-x/\xi)$ gives
\begin{align*}
    V_\Pi^{\lambda,\xi}(f) \;=\; -\xi\log\Big[\tfrac14 L_1(v)^r+\tfrac34 L_2(v)^r\Big],
\end{align*}
while, for a candidate $q$, $-\theta'\log\mathbb{E}_q[\exp(-u(f)/\theta')] = -\theta'\log[\sum_\omega q(\omega)v_\omega^s]$.

Suppose some $(\theta',q)$ had the stated property. Hold $v_\omega\equiv\kappa$ for every $\omega\ne\omega_1$, for a fixed $\kappa$ in the range swept out by acts, and let $v_{\omega_1}=\kappa e^{t}$ with $t$ varying over an open interval around $0$; this is admissible because $\mathcal{W}$ contains an open box. Writing $A(t)$ and $B(t)$ for the resulting bracketed expressions, the hypothesised equality reduces, after the $\kappa$-dependent constants cancel, to
\begin{equation}
    s\log A(t) \;=\; r\log B(t)
    \label{eq:mismatch-identity-body}
\end{equation}
for all $t$ close to 0. Both sides are smooth and vanish at $t=0$, so every derivative at $t=0$ must agree. Appendix \ref{app:pf-mismatch} carries out the first three. The first forces $q(\omega_1)=\bar p=\tfrac38$, independently of $r$ and $s$. The second then forces
\begin{align*}
    s \;=\; 1+(r-1)\cdot\frac{\sigma^2_\mu(p)}{\bar p(1-\bar p)} \;=\; 1+\frac{r-1}{5} \;=\; \frac{r+4}{5},
\end{align*}
using $\sigma^2_\mu(p)=\tfrac14\cdot\tfrac34\cdot(\tfrac12)^2=\tfrac3{64}$ and $\bar p(1-\bar p)=\tfrac{15}{64}$; note $s>0$ for every $r>0$, so this does determine a candidate $\theta'=5\lambda/(r+4)$ rather than an immediate contradiction. The third derivative, evaluated at these forced values, is
\begin{align*}
    \frac{d^3}{dt^3}\Big[s\log A(t)-r\log B(t)\Big]_{t=0}
    \;=\; \frac{3\,r\,(r-1)^2(r+4)}{1600},
\end{align*}
which is strictly positive for every $r>0$ with $r\ne1$. This contradicts \eqref{eq:mismatch-identity-body}, so no such $(\theta',q)$ exists. Thus, the final sentence of the theorem is immediate, since a $\theta'$ working for every $(\Pi,\mu)$ would in particular work for this one.
\end{proof}


\subsection{Proof of Lemma \ref{lem:power-tilt}}
\label{app:pf-power-tilt}

\begin{proof}
Let $r:=p/q$, and minimize $\sum_\omega q(\omega)[r(\omega)w(\omega)+\theta f_\gamma(r(\omega))]$ subject to $\sum_\omega q(\omega)r(\omega)=1$. 
Each $f_\gamma$ is strictly convex on $(0,\infty)$ and the constraint set is convex, so a minimizer, if interior, is unique and characterized by the first-order condition $f_\gamma'(r(\omega))=(\nu-w(\omega))/\theta$ for a multiplier $\nu$. 
Interiority is where the restriction on $\theta$ enters. Writing $a:=\gamma-1$, the derivative $f_\gamma'(t)=(t^a-1)/a$ has range $\mathbb{R}$ for $\gamma=1$, $(-\infty,1/(1-\gamma))$ for $\gamma<1$, and $(-1/(\gamma-1),\infty)$ for $\gamma>1$. For $\gamma\le1$ the constraint $\mathbb{E}_q[r]=1$ always admits a root with $r>0$, since $f_\gamma'(0^+)=-\infty$. 
For $\gamma>1$ the generator extends continuously to $t=0$ with $f_\gamma(0)=1/\gamma<\infty$, so the boundary of the simplex carries only a finite penalty and the first-order condition is solvable only if $1+(\gamma-1)(\nu-w(\omega))/\theta>0$ for every $\omega$; evaluating at $\nu=\max_\omega w(\omega)-\theta/(\gamma-1)$ shows that $\theta>(\gamma-1)\mathrm{osc}(w)$ suffices. 
Without it the infimum over $\Delta_{++}(\Omega)$ need not be attained --- for $\Omega=\{\omega_1,\omega_2\}$, $q$ uniform, $w=(1,-1)$ and $\gamma=2$, the minimizer over $\Delta(\Omega)$ is the vertex $(0,1)$ once $\theta\le1$. 
For $\gamma\ne1$, $f_\gamma'(t)=(t^{\gamma-1}-1)/(\gamma-1)$, which inverts to $r(\omega)=[1+(\gamma-1)(\nu-w(\omega))/\theta]^{1/(\gamma-1)}$. 
For $\gamma=1$, $f_1'(t)=\log t$ gives $r(\omega)=\exp((\nu-w(\omega))/\theta)$, and the constraint pins $\exp(\nu/\theta)=1/\mathbb{E}_q[\exp(-w/\theta)]$, reproducing the Gibbs tilt of Lemma \ref{lem:dv}.
\end{proof}

\subsection{Proof of Proposition \ref{prop:kl-uniqueness}}
\label{app:pf-kl}

\begin{proof}
We prove (i) -- (iv), respectively.
\paragraph{Proof of (i)}
(i) is Lemma \ref{lem:dv} with $u(f)$ replaced by $w$.

\paragraph{Proof of (ii)}
$D_{\chi^2}=2D_2$ corresponds to the generator $2f_2(t)=(t-1)^2$, with $2f_2'(t)=2(t-1)$, so Lemma \ref{lem:power-tilt}'s first-order condition reads $r(\omega)=1+(\nu-w(\omega))/(2\theta)$, and $\sum_\omega q(\omega)r(\omega)=1$ forces $\nu=\bar w$. The hypothesis $\max_\omega w(\omega)-\bar w<2\theta$ is exactly what makes the resulting $r(\omega)=1-(w(\omega)-\bar w)/(2\theta)$ strictly positive at every state, so that $p=rq$ is interior. Substituting gives $\mathbb{E}_p[w]=\bar w-\kappa_2/(2\theta)$ and $\theta D_{\chi^2}(p\|q)=\kappa_2/(4\theta)$, which sum to the stated value. On this range $r$ is affine in $w$, so no higher moment of $w$ enters at any order in $\theta$.

\paragraph{Proof of (iii)}
For $\gamma=1$, expand $\log\mathbb{E}_q[\exp(-w/\theta)] = \sum_{n\ge1}\kappa_n(-1/\theta)^n/n!$ in the cumulants of $w$ under $q$ and multiply by $-\theta$; the coefficient of $\theta^{-2}$ is $\kappa_3/6=(2-1)\kappa_3/6$, as claimed. For $\gamma\ne1$, set $\varepsilon:=1/\theta$ and $a:=\gamma-1$. By Lemma \ref{lem:power-tilt}, $r(\omega)=(1+a\delta(\omega))^{1/a}$ with $\delta(\omega):=\varepsilon(\nu-w(\omega))$; write $\tilde w:=w-\bar w$ and substitute $\nu=\bar w+\varepsilon c$. The existence and regularity of the root $c(\varepsilon)$ follow from the implicit function theorem applied not to $G(\nu,\varepsilon):=\mathbb{E}_q[r]-1$ --- which vanishes identically at $\varepsilon=0$, so that $\partial_\nu G(\cdot,0)\equiv0$ --- but to $H(c,\varepsilon):=\varepsilon^{-2}(\mathbb{E}_q[r]-1)$, which extends analytically to $\varepsilon=0$ with $H(c,0)=c+\tfrac{1-a}{2}\kappa_2$ and $\partial_cH(c,0)=1\ne0$. Writing $c_1:=c(0)$, so that $\nu=\bar w+c_1\varepsilon+O(\varepsilon^2)$, we have $\delta(\omega)=-\varepsilon\tilde w(\omega)+c_1\varepsilon^2+O(\varepsilon^3)$. Expanding $(1+a\delta)^{1/a}=\exp[\tfrac1a\log(1+a\delta)]$ in powers of $\delta$ gives $r=1+\delta+\tfrac{1-a}2\delta^2+\tfrac{(a-1)(2a-1)}6\delta^3+O(\delta^4)$; imposing $\mathbb{E}_q[r]=1$ order by order in $\varepsilon$ (using $\mathbb{E}_q[\tilde w]=0$) leaves $c_1$ unconstrained at orders $\varepsilon^0$ and $\varepsilon^1$ and yields $c_1=\tfrac{a-1}{2}\kappa_2$ at order $\varepsilon^2$. With $x:=r-1$, Taylor expansion of $f_\gamma$ at $t=1$ --- using $f_\gamma''(1)=1$ and $f_\gamma'''(1)=\gamma-2$, both computed directly from \eqref{eq:power-divergence} --- gives $f_\gamma(1+x)=\tfrac{x^2}2+\tfrac{\gamma-2}6x^3+O(x^4)$. Composing with $x=\delta+\tfrac{2-\gamma}2\delta^2+O(\delta^3)$ and $\delta=-\varepsilon\tilde w+c_1\varepsilon^2+O(\varepsilon^3)$, then taking $\mathbb{E}_q[\cdot]$ of both $\mathbb{E}_p[w]=\bar w+\mathbb{E}_q[x\tilde w]$ and $\theta\,\mathbb{E}_q[f_\gamma(r)]=\varepsilon^{-1}\mathbb{E}_q[f_\gamma(1+x)]$ and collecting through $O(\varepsilon^2)$ (the $c_1$-dependent terms cancel by $\mathbb{E}_q[\tilde w]=0$), gives the stated expansion.

The formula passes two independent checks. At $\gamma=1$ it reproduces the cumulant expansion just derived, with coefficient $(2-1)/6=1/6$. At $\gamma=2$ the $\kappa_3$ coefficient vanishes identically, matching (ii)'s exact quadratic formula after the normalization $D_{\chi^2}=2f_2$, i.e.\ $\theta\mapsto\theta/2$.

\paragraph{Proof of (iv)}
Suppose $\Lambda_\gamma(\,\cdot\,;\theta)=\Lambda_1(\,\cdot\,;\vartheta)$ on all of $\mathcal{W}$ for some $\gamma\ne1$, with $\theta$ and $\vartheta$ \emph{fixed}. 
The two expansions in (iii) are in different variables, so they cannot be compared directly; the comparison is made available by the exact scaling identity
\begin{align*}
    \Lambda_\gamma(tw;\theta) \;=\; \min_p\big\{t\,\mathbb{E}_p[w]+\theta D_\gamma(p\|q)\big\} \;=\; t\,\Lambda_\gamma\big(w;\theta/t\big),
    \qquad t>0,
\end{align*}
together with the fact that $\mathcal{W}$ contains a box about the origin whenever $\underline u<0<\bar u$, so that $tw\in\mathcal{W}$ for all $t\in(0,1)$ and all $w$ in that box.\footnote{If $0\notin(\underline u,\bar u)$, replace $u$ by $u-c$ for a constant $c$ in its range; both $\Lambda_\gamma$ and $\Lambda_1$ shift by $c$, so the claimed identity is unaffected.} Applying the hypothesised identity at $tw$ and dividing by $t$ gives $\Lambda_\gamma(w;\theta/t)=\Lambda_1(w;\vartheta/t)$ for every $t\in(0,1)$. Writing $\Theta:=\theta/t$ and $\beta:=\vartheta/\theta$, both sides are now expansions in the \emph{single} variable $\Theta\to\infty$:
\begin{align*}
    \bar w-\frac{\kappa_2}{2\Theta}+\frac{2-\gamma}{6}\cdot\frac{\kappa_3}{\Theta^2}+O(\Theta^{-3})
    \;=\;
    \bar w-\frac{\kappa_2}{2\beta\Theta}+\frac{1}{6}\cdot\frac{\kappa_3}{\beta^2\Theta^2}+O(\Theta^{-3}).
\end{align*}
Matching at order $\Theta^{-1}$ with $\kappa_2\ne0$ forces $\beta=1$, i.e.\ $\vartheta=\theta$; matching at order $\Theta^{-2}$ with $\kappa_3\ne0$ then forces $(2-\gamma)/6=1/6$, i.e.\ $\gamma=1$, a contradiction. The argument lives entirely at large $\Theta$, where Lemma \ref{lem:power-tilt}'s tilt is interior for every $\gamma$. It therefore suffices that $\mathcal{W}$ contain some $w$ with $\kappa_2(w)\ne0$ and $\kappa_3(w)\ne0$. For a two-valued profile taking $w_1$ with probability $\varrho$ and $w_2$ with probability $1-\varrho$, writing $d:=w_1-w_2$, one computes $\kappa_2=\varrho(1-\varrho)d^2$ and $\kappa_3=\varrho(1-\varrho)(1-2\varrho)d^3$, so both are nonzero exactly when $d\ne0$ and $\varrho\ne1/2$. If $|\Omega|=2$ with $q=(q_1,1-q_1)$ and $q_1\ne1/2$, any non-constant $w\in\mathcal{W}$ will do. If $|\Omega|\ge3$, at most one coordinate of $q$ can equal $1/2$, so choose $\omega_i$ with $q(\omega_i)\ne1/2$ and set $w:=c\mathbf{1}\{\omega=\omega_i\}$ for a small $c\ne0$; this lies in $\mathcal{W}$ by Section \ref{sec:primitives} and is two-valued with $\varrho=q(\omega_i)$.
\end{proof}

\subsection{Proof of Proposition \ref{prop:ri-gap}}
\label{app:pf-ri-gap}

In \eqref{eq:ri-degenerate} the decision-maker chooses an information structure, that is, a Bayes-plausible distribution over posteriors $\tau\in\Delta(\Delta_{++}(\Omega))$ with $\mathbb{E}_\tau[p]=q$, and trades the value of what she learns against the cost $\theta\,\mathbb{E}_\tau[D_{\mathrm{KL}}(p\|q)]$ of \citet{S_2003} and \citet{PST_2023}. The cost is \emph{uniformly posterior separable}: the same function $\theta D_{\mathrm{KL}}(p\|q)$ is evaluated at each realized posterior and averaged under $\tau$, which is the standard specification and is what makes \eqref{eq:ri-degenerate} a genuine information-acquisition problem, if a degenerate one. The distribution $\tau$ over first-order beliefs should not be confused with the second-order prior $\mu\in\Delta(\Pi)$.

\begin{proof}[Proof of Proposition \ref{prop:ri-gap}]
For any Bayes-plausible $\tau$ (i.e.\ $\mathbb{E}_\tau[p]=q$), the map $p\mapsto\mathbb{E}_p[u(f)]$ is linear, so
\begin{align*}
    \mathbb{E}_\tau\big[\mathbb{E}_p[u(f)]\big] \;=\; \mathbb{E}_{\mathbb{E}_\tau[p]}[u(f)] \;=\; \mathbb{E}_q[u(f)]:
\end{align*}
the objective's first term is the \emph{same constant for every Bayes-plausible $\tau$}. Hence \eqref{eq:ri-degenerate} reduces to $\mathbb{E}_q[u(f)] - \theta\min_\tau \mathbb{E}_\tau[D_{\mathrm{KL}}(p\|q)]$, and since $D_{\mathrm{KL}}(p\|q)\ge0$ with equality iff $p=q$, the minimum over Bayes-plausible $\tau$ is $0$, uniquely at $\tau=\delta_q$; this proves the first claim.

For the inequalities, Lemma \ref{lem:dv} gives
\begin{align*}
    V(f) \;=\; \min_p\big\{\mathbb{E}_p[u(f)]+\theta D_{\mathrm{KL}}(p\|q)\big\}
    \;\le\; \mathbb{E}_q[u(f)]+\theta D_{\mathrm{KL}}(q\|q)
    \;=\; \mathbb{E}_q[u(f)].
\end{align*}
The gap is exactly computable, so no separate convexity argument is needed. The proof of Lemma \ref{lem:dv} establishes the identity $\mathbb{E}_p[u(f)]+\theta D_{\mathrm{KL}}(p\|q)-V(f)=\theta D_{\mathrm{KL}}(p\|p^*)$ for every $p\ll q$, where $p^*(\omega)\propto q(\omega)\exp(-u(f(\omega))/\theta)$ is the Gibbs minimizer; evaluating it at $p=q$ gives
\begin{align*}
    \mathbb{E}_q[u(f)] - V(f) \;=\; \theta D_{\mathrm{KL}}(q\|p^*),
\end{align*}
which is strictly positive unless $q=p^*$, that is, unless $u(f)$ is constant on $\mathrm{supp}(q)$. Symmetrically, applying the same argument to $-u(f)$ and using
\begin{align*}
    W_{\mathrm{RI}}(f) \;=\; -\Big(\min_p\big\{\mathbb{E}_p[-u(f)]+\theta D_{\mathrm{KL}}(p\|q)\big\}\Big)
\end{align*}
gives
\begin{align*}
    W_{\mathrm{RI}}(f)-\mathbb{E}_q[u(f)] \;=\; \theta D_{\mathrm{KL}}(q\|p^\circ),
    \qquad p^\circ(\omega)\propto q(\omega)\exp\big(u(f(\omega))/\theta\big),
\end{align*}
so $W_{\mathrm{RI}}(f)\ge \mathbb{E}_q[u(f)]$, strictly whenever $u(f)$ is non-constant on $\mathrm{supp}(q)$.
\end{proof}

\subsection{Proof of Proposition \ref{prop:degeneracy}}
\label{app:degeneracy}

Throughout this subsection we use the exact form Bayes' rule takes for i.i.d.\ sampling. Suppose there is a true model $\pi^\dagger$ generating an i.i.d.\ sequence $\omega_1,\omega_2,\ldots \sim \pi^\dagger$ on the product space $(\Omega^\infty,\mathcal{B},(\pi^\dagger)^{\otimes\infty})$,\footnote{Here $(\pi^\dagger)^{\otimes\infty}$ denotes the countably infinite product measure: $\pi^{\otimes N}$ is the law of $N$ independent draws from $\pi$, and $\pi^{\otimes\infty}$ is the law of the whole sequence $\omega_1,\omega_2,\ldots$ on $(\Omega^\infty,\mathcal{B})$, well defined by Kolmogorov's extension theorem. We carry the superscript because the almost-sure statements below concern sample paths, and $\pi^\dagger$ itself is a measure on $\Omega$ rather than on $\Omega^\infty$.} and a Bayesian with prior $\mu_0\in\Delta(\Pi)$ forms the posterior $\mu_N(\cdot\mid\omega^N)$ in the usual way after observing $\omega^N=(\omega_1,\ldots,\omega_N)$. Writing $L_N := \frac1N\sum_{i=1}^N\delta_{\omega_i}$ for the empirical distribution and using
\begin{align*}
    \sum_\omega L_N(\omega)\log\pi(\omega) \;=\; -H(L_N)-D_{\mathrm{KL}}(L_N\|\pi),
\end{align*}
Bayes' rule reduces to the exact identity
\begin{equation}
    \mu_N(\pi\mid\omega^N) = \frac{\mu_0(\pi)\,\exp(-ND_{\mathrm{KL}}(L_N\|\pi))}{\int_\Pi \mu_0(d\pi')\,\exp(-ND_{\mathrm{KL}}(L_N\|\pi'))}, \qquad \forall N,\ \forall\omega^N,
    \label{eq:bayes-gibbs}
\end{equation}
with $L_N\to\pi^\dagger$ $\pi^\dagger$-a.s.\ by the Strong Law of Large Numbers.

\begin{proof}
\emph{Necessity of the support condition.} For every $N$, \eqref{eq:bayes-gibbs} exhibits $\mu_N$ as a reweighting of $\mu_0$ by a strictly positive density, so $\mu_N\ll\mu_0$ and hence $\mu_N(\mathrm{supp}\,\mu_0)=1$. The set $\Pi\setminus\mathrm{supp}\,\mu_0$ is open in $\Pi$; if it contained $\pi^\dagger$, the portmanteau theorem would require $\liminf_N\mu_N(\Pi\setminus\mathrm{supp}\,\mu_0)\ge\delta_{\pi^\dagger}(\Pi\setminus\mathrm{supp}\,\mu_0)=1$, whereas that probability is $0$ for every $N$.

\emph{Sufficiency.} Since $\Pi$ is compact and contained in $\Delta_{++}(\Omega)$, $c:=\min_{\pi\in\Pi}\min_{\omega\in\Omega}\pi(\omega)>0$. Consider the normalized log-likelihood-ratio process
\begin{align*}
    \Lambda_N(\pi) \;:=\; \frac1N\log\frac{\mathbb{P}_{\pi^\dagger}(\omega^N)}{\mathbb{P}_{\pi}(\omega^N)}
    \;=\; \sum_\omega L_N(\omega)\log\frac{\pi^\dagger(\omega)}{\pi(\omega)}
    \;=\; D_{\mathrm{KL}}(L_N\|\pi)-D_{\mathrm{KL}}(L_N\|\pi^\dagger).
\end{align*}
The two exponents in \eqref{eq:bayes-gibbs} differ by the $\pi$-free quantity $D_{\mathrm{KL}}(L_N\|\pi^\dagger)$, which cancels between numerator and normalizing constant, so $\mu_N(d\pi)\propto\exp(-N\Lambda_N(\pi))\,\mu_0(d\pi)$ exactly. Because $\pi^\dagger,\pi\in\Pi\subseteq[c,1]^\Omega$ gives $|\log(\pi^\dagger(\omega)/\pi(\omega))|\le\log(1/c)$,
\begin{equation}
\begin{split}
    \sup_{\pi\in\Pi}\big|\Lambda_N(\pi)-D_{\mathrm{KL}}(\pi^\dagger\|\pi)\big|
    &\;=\; \sup_{\pi\in\Pi}\Big|\sum_\omega\big(L_N(\omega)-\pi^\dagger(\omega)\big)\log\tfrac{\pi^\dagger(\omega)}{\pi(\omega)}\Big| \\
    &\;\le\; \log(1/c)\,\|L_N-\pi^\dagger\|_1 \;=:\; \eta_N,
\end{split}
    \label{eq:uniform-llr}
\end{equation}
and $\eta_N\to0$ $(\pi^\dagger)^{\otimes\infty}$-a.s.\ by the Strong Law of Large Numbers on the finite alphabet $\Omega$. Display \eqref{eq:uniform-llr} is the precise, \emph{uniform} content of the statement that the exponent in \eqref{eq:bayes-gibbs} concentrates at $\pi^\dagger$; note that at finite $N$ the exponent is in general maximized not at $\pi^\dagger$ but at the reverse information projection of $L_N$ onto $\Pi$, so the uniformity is what the argument actually needs.

For Borel $A,B\subseteq\Pi$ with $\mu_0(B)>0$, bounding the numerator above and the normalizing constant below using \eqref{eq:uniform-llr} gives
\begin{align*}
    \mu_N(A) \;\le\; \frac{1}{\mu_0(B)}\,
    \exp\!\Big(\!-N\Big[\inf_A D_{\mathrm{KL}}(\pi^\dagger\|\cdot)-\sup_B D_{\mathrm{KL}}(\pi^\dagger\|\cdot)-2\eta_N\Big]\Big).
\end{align*}
Fix $\varepsilon>0$ and put $A:=\Pi\setminus B_\varepsilon(\pi^\dagger)$. The map $D_{\mathrm{KL}}(\pi^\dagger\|\cdot)$ is lower semicontinuous on the compact set $A$ and strictly positive there by Gibbs' inequality (which also supplies the uniqueness of its zero), so $2\delta:=\inf_A D_{\mathrm{KL}}(\pi^\dagger\|\cdot)>0$; by continuity at $\pi^\dagger$ choose $\rho\in(0,\varepsilon)$ with $\sup_{B_\rho(\pi^\dagger)}D_{\mathrm{KL}}(\pi^\dagger\|\cdot)\le\delta/2$ and put $B:=B_\rho(\pi^\dagger)$, so that $\mu_0(B)>0$ precisely because $\pi^\dagger\in\mathrm{supp}\,\mu_0$. On the almost-sure event $\{\eta_N\to0\}$ we eventually have $2\eta_N\le\delta/2$, whence $\mu_N(\Pi\setminus B_\varepsilon(\pi^\dagger))\le\mu_0(B_\rho(\pi^\dagger))^{-1}\exp(-N\delta)\to0$. As $\varepsilon>0$ was arbitrary and $\Pi$ is a compact metric space, portmanteau gives $\mu_N\Rightarrow\delta_{\pi^\dagger}$. The Strong Law supplies a \emph{single} null set, independent of $\varepsilon$ and of the integrand, so the exceptional set does not depend on $f$.

Finally, $\pi\mapsto\phi(\mathbb{E}_\pi[u(f)])$ is bounded and continuous on $\Pi$ for each $f\in\mathcal{F}$, so weak convergence delivers the displayed limit.
\end{proof}

The argument above uses no Laplace expansion, and in particular requires neither a Lebesgue density for $\mu_0$ nor interiority of $\pi^\dagger$ in $\Pi$: what is needed is only the first-order, exponential-rate comparison \eqref{eq:uniform-llr}, i.e.\ a Laplace \emph{principle} in the sense of \citet[Ch.~1]{DE_1997} rather than a Laplace \emph{method}. The identity behind \eqref{eq:bayes-gibbs} is the method of types \citep[\S11.1]{CT_2006}; \citet{GO_1999} establish, for exactly this finite-alphabet setting, a large-deviation principle for $\mu_N$ along any sample path whose empirical distribution converges to a point of $\mathrm{supp}\,\mu_0$, with good rate function $\nu\mapsto D_{\mathrm{KL}}(\pi^\dagger\|\nu)$ on that support and $+\infty$ off it --- so that the support condition enters both as a hypothesis and as the effective domain of the rate function. They emphasize that the finiteness of $\Omega$ is what allows the statement to dispense with further conditions on the prior, since beyond finite alphabets the support condition is no longer sufficient. The support condition is \citeauthor{Sch_1965}'s (\citeyear{Sch_1965}) Kullback--Leibler support condition in the present setting: with $\Pi$ compact in $\Delta_{++}(\Omega)$, the bound $\log x\le x-1$ and Pinsker's inequality give
\begin{align*}
    \tfrac12\|\pi^\dagger-\pi\|_1^2 \;\le\; D_{\mathrm{KL}}(\pi^\dagger\|\pi) \;\le\; \|\pi^\dagger-\pi\|_1/c ,
\end{align*}
so Kullback--Leibler neighborhoods and Euclidean ones are mutually cofinal and the two support notions coincide. For pathwise concentration bounds of this kind in an economic setting, including priors without full support, see \citet{FLS_2023}.

\paragraph{The misspecified case.} Under misspecification --- that is, when the data are generated by some $\pi^*$ not belonging to $\mathrm{supp}\,\mu_0$ --- the uniform estimate \eqref{eq:uniform-llr} above shows that $\mu_N$ concentrates instead on
\begin{align*}
    \argmin_{\pi\in\mathrm{supp}\,\mu_0} D_{\mathrm{KL}}(\pi^*\|\pi),
\end{align*}
the information projection of the true law $\pi^*$ onto the support of the prior --- the classical fact of \citet{B_1966}.\footnote{Apply \eqref{eq:uniform-llr} to $-\sum_\omega L_N(\omega)\log\pi(\omega)$ in place of $\Lambda_N$; the additive entropy term $H(\pi^*)$ does not depend on $\pi$ and again cancels between numerator and normalizing constant.} Two qualifications apply. The minimization is over $\mathrm{supp}\,\mu_0$ and not over $\Pi$, for the absolute-continuity reason given in the proof above; and when that argmin is not a singleton the posterior need not converge at all, though it does concentrate on the set. Uniqueness therefore requires a further condition; convexity of $\mathrm{supp}\,\mu_0$ suffices provided $\pi^*\in\Delta_{++}(\Omega)$, since $\pi\mapsto D_{\mathrm{KL}}(\pi^*\|\pi)$ is then strictly convex on $\Delta_{++}(\Omega)$. Full support of $\pi^*$ cannot be dropped: for $\pi^*=(\tfrac12,\tfrac12,0,0)$ the map is constant along the segment joining $(0.4,0.4,0.15,0.05)$ to $(0.4,0.4,0.05,0.15)$, so a convex support can still leave the minimizer non-unique. This limiting belief is the belief condition of a Berk--Nash equilibrium in the sense of \citet{EP_2016}, specialized to passive learning: because the decision-maker's actions here do not affect the data-generating process, the Kullback--Leibler objective is exogenous and no fixed point in (action, belief) arises.

\subsection{Relation to \citet{PST_2023}}
\label{app:pst}

\begin{remark}[Relation to \citet{PST_2023}]
\label{rem:pst-gap}
Proposition \ref{prop:ri-gap} already shows, at a categorical level, that no cost function $C(p,q)$ defined on single realized posteriors can be \emph{derived} as the value of a genuine information-acquisition problem about a fixed act; the argument there depends on linearity and Bayes-plausibility, not on which $C$ is postulated. A second point is specific to \citet{PST_2023}: their axiomatization does not deliver $C(p,q)=\theta D_{\mathrm{KL}}(p\|q)$ as a special case in the first place. Their Theorem 1 characterizes the cost of an \emph{experiment} $\sigma=(\sigma_i)_{i\in\Theta}$, under invariance across Blackwell-equivalent experiments, additivity across independent experiments, dilution linearity and continuity, as $C(\sigma)=\sum_{i\ne j}\beta_{ij}D_{\mathrm{KL}}(\sigma_i\|\sigma_j)$ for a unique nonnegative $(\beta_{ij})$. In the two-state case with unit weights this is, as the authors note, the $J$-divergence $D_{\mathrm{KL}}(\sigma_1\|\sigma_2)+D_{\mathrm{KL}}(\sigma_2\|\sigma_1)$ of Jeffreys --- a symmetrized object, not the single directed term that \eqref{eq:wri} uses. Their Section VI shows that such a cost is uniformly posterior separable exactly when $\beta_{ij}(q)=b_{ij}q_i$ for prior-independent constants $b_{ij}$. In that case it averages $F(p)-F(q)$ over the posteriors induced by $\sigma$, with $F(p)=\sum_{i\ne j}b_{ij}\,p_i\log(p_i/p_j)$ (their Eq.~(16)). Since the constants $b_{ij}$ cannot depend on $q$, the prior never enters $F$ inside a logarithm, and no member of the class reduces to $\theta D_{\mathrm{KL}}(p\|q)$ for a single realized $p$. The taxonomy of \citet[Ch.~6]{Str_2025} locates the difficulty: their class is prior-independent, hence posterior separable but not uniformly so, and the one nearby cost that does average a directed divergence from the prior is mutual information, which they exclude because it is subadditive rather than additive across independent experiments. The cost in \eqref{eq:wri} is therefore best read as a distinct, simpler postulate --- standard in the robust-control tradition following \citet{HS_2001} and axiomatized at the preference level by \citet{S_2011} --- motivated by, but not a corollary of, that experiment-level result. Which further restriction on experiments --- plausibly to one-directional, perfectly informative binary ones --- would make it specialize exactly to $\theta D_{\mathrm{KL}}(p\|q)$ is left open.
\end{remark}

\subsection{The Two-Point Reduction Behind Remark \ref{rem:binary-richness}}
\label{app:binary-richness}

This subsection verifies the claim of Remark \ref{rem:binary-richness}: under binary richness, Theorem \ref{thm:necessity} holds verbatim on an arbitrary measurable state space, and Theorem \ref{thm:mismatch-impossibility} holds in the uniform form stated as its second conclusion.

Let $\mathcal{G}:=\{\emptyset,A,A^c,\Omega\}$ and restrict attention to $\mathcal{G}$-measurable acts, that is, to $f=\ell_1\mathbf{1}_A+\ell_2\mathbf{1}_{A^c}$ with $\ell_1,\ell_2\in\Delta(X)$. For such an act $u(f)$ takes the two values $w_1:=u(\ell_1)$ and $w_2:=u(\ell_2)$, and $(w_1,w_2)$ sweeps the open square $(\underline u,\bar u)^2$ as $\ell_1,\ell_2$ vary independently. By Lemma \ref{lem:dv},
\begin{align*}
    v^\lambda_\pi(f) \;=\; -\lambda\log\Big(\pi(A)\exp(-w_1/\lambda)+\big(1-\pi(A)\big)\exp(-w_2/\lambda)\Big),
\end{align*}
which depends on $\pi$ only through the scalar $\pi(A)$; identically, $V(f)$ depends on the baseline only through $q(A)$, and the predictive baseline satisfies $\bar\pi(A)=m\pi_1(A)+(1-m)\pi_2(A)$. Hence, for $\Pi=\{\pi_1,\pi_2\}$ and $\mu=(m,1-m)$, the hypothesis of either theorem restricted to $\mathcal{G}$-measurable acts coincides exactly with the corresponding hypothesis on the two-point state space $\{A,A^c\}$, carrying the models $(\pi_i(A),1-\pi_i(A))\in\Delta_{++}(\{A,A^c\})$ and the same $\mu$. The proofs of Theorems \ref{thm:necessity} and \ref{thm:mismatch-impossibility} apply verbatim to that two-point problem. For Theorem \ref{thm:necessity} the conclusion is again obtained on all of $(\underline u,\bar u)$, since $w_1=w_2=c$ gives $v^\theta_\pi(f)=c$ for every $c$ in that interval.

For Theorem \ref{thm:mismatch-impossibility} one qualification is needed. 
Binary richness requires a \emph{single} pair of models, whereas the proof of Theorem \ref{thm:mismatch-impossibility} \emph{selects} $\pi_1(\omega_1)=\tfrac34$ and $\pi_2(\omega_1)=\tfrac14$ in order to make the third-order obstruction explicit. 
With the models given rather than chosen, what survives verbatim is the uniform conclusion --- that no single $\theta'$ works across all priors --- and it survives at second order alone. 
Letting $a:=\pi_1(A)$, $b:=\pi_2(A)$ and $\bar p:=ma+(1-m)b$, one has $\sigma^2_\mu(p)=m(1-m)(a-b)^2$, so \eqref{eq:s-formula} reads $s=1+(r-1)R(m)$ with
\begin{align*}
    R(m) \;=\; \frac{m(1-m)(a-b)^2}{\bar p\,(1-\bar p)}.
\end{align*}
As $m\to0^+$ we have $\bar p\to b\in(0,1)$ and hence $R(m)\to0$, while $R(\tfrac12)=(a-b)^2/\big(4\bar p(1-\bar p)\big)>0$ because $a\ne b$. Since $r\ne1$, two values of $m$ already force two values of $s$: either some $s$ is non-positive, in which case no admissible $\theta'=\lambda/s$ exists for that prior, or the two priors demand two different positive $\theta'$. 
Either way no single $\theta'$ serves both. Whether the \emph{stronger} single-prior conclusion also holds depends on the given pair $(a,b)$, since the third derivative can degenerate for particular values; we do not pursue that here. (When $r=1$ the same formula gives $s=1$ and $\theta'=\lambda$ for every $m$, consistently with Theorem \ref{thm:collapse}(b).)

\section{A Menu-Based Comparison with Genuine Rational Inattention}
\label{app:ri-menu}

Section \ref{sec:ri-gap} shows that $V(f)$ and $W_{\mathrm{RI}}(f)$ --- despite the latter's suggestive name --- are not values of any genuine information-acquisition problem about the single fixed act $f$: they are, respectively, the robust-control (\citealp{HS_2001}; \citealp{S_2011}) and Wishful Thinking (\citealp{CL_2019}; \citeauthor{RSS_2023}, \citeyear{RSS_2023}, Prop.~1) faces of one Donsker--Varadhan identity. 
Genuine Rational Inattention (\citealp{S_2003}; \citealp{MM_2015}) only becomes a non-trivial problem once there is a real \emph{menu} of acts for a stochastic choice rule to match to the acquired signal --- Proposition \ref{prop:ri-gap}'s degeneracy argument is precisely the observation that a single act leaves nothing for a choice rule to do. 
This appendix makes that comparison precise: it defines the genuine menu-level Rational Inattention value, gives it an exact closed form via Lemma \ref{lem:dv}, and shows that $V$ and $W_{\mathrm{RI}}$, maximized over the same menu, bound it from either side.

\subsection{The Genuine Rational Inattention Value}

Fix a finite menu of acts $F=\{f_1,\dots,f_K\}$ and a baseline prior $q\in\Delta_{++}(\Omega)$. A \emph{stochastic choice rule} is a map $\rho:\Omega\to\Delta(F)$, written $\rho(f\mid\omega)$, interpreted as the decision-maker choosing act $f\in F$ with probability $\rho(f\mid\omega)$ after observing a signal correlated with the true state $\omega$ (the signal itself is not modeled explicitly, following \citet{MM_2015}: only the joint distribution $q(\omega)\rho(f\mid\omega)$ over $(\omega,f)$ that some signal could induce matters for payoffs and for cost, since attention cost is posited to depend only on this joint distribution). Write $\rho(f):=\sum_{\omega}q(\omega)\rho(f\mid\omega)$ for the induced (unconditional) choice probability of $f$, and
$$
I(\rho) := \sum_{\omega\in\Omega} q(\omega)\, D_{\mathrm{KL}}\big(\rho(\cdot\mid\omega)\,\big\|\,\rho(\cdot)\big)
$$
for the mutual information between the state and the chosen act under $\rho$. The genuine Rational Inattention value of the menu $F$ at cost coefficient $\theta>0$ is
\begin{equation}
    \mathcal{RI}(q;F,\theta) := \max_{\rho:\,\Omega\to\Delta(F)} \left\{ \sum_{\omega\in\Omega} q(\omega) \sum_{f\in F} \rho(f\mid\omega)\, u(f(\omega)) \;-\; \theta\, I(\rho) \right\}.
    \label{eq:ri-menu}
\end{equation}
Unlike $\mathcal{V}^{\deg}_{\mathrm{RI}}(q)$ in \eqref{eq:ri-degenerate}, \eqref{eq:ri-menu} is a genuine choice problem: because $u(f(\omega))$ varies across $f\in F$, a choice rule that shifts probability toward the state's best act as $\omega$ varies raises the first term, and the mutual-information cost is exactly what makes such state-contingent adjustment costly rather than free. When $|F|=1$, $I_\rho\equiv0$ for the only feasible $\rho$ (there is nothing to condition choice on), and \eqref{eq:ri-menu} collapses to $\mathbb{E}_q[u(f)] = \mathcal{V}^{\deg}_{\mathrm{RI}}(q)$, consistently with Proposition \ref{prop:ri-gap}.

\subsection{An Exact Reformulation via Lemma \ref{lem:dv}}

The value \eqref{eq:ri-menu} has no closed form as it stands, because $I(\rho)$ is a nonlinear function of $\rho$ through its own marginal $\rho(\cdot)$. The next lemma removes this circularity by writing mutual information as a minimum, over an \emph{external} reference distribution $\alpha\in\Delta(F)$, of an average KL divergence --- a standard information-theoretic identity (the ``golden formula''; see, e.g., \citealp[Thm.~2.4.3]{CT_2006}) --- which we verify directly for completeness.

\begin{lemma}[Mutual Information as a Minimal Average Divergence]
\label{lem:mi-variational}
For every $\rho:\Omega\to\Delta(F)$,
$$
I(\rho) = \min_{\alpha\in\Delta(F)} \sum_{\omega\in\Omega} q(\omega)\, D_{\mathrm{KL}}\big(\rho(\cdot\mid\omega)\,\big\|\,\alpha\big),
$$
with the minimum attained uniquely at $\alpha=\rho(\cdot)$, the $\rho$-induced marginal.
\end{lemma}
\begin{proof}
If $\rho(\cdot\mid\omega)\not\ll\alpha$ for some $\omega$ then both sides below are $+\infty$ and such $\alpha$ cannot be minimizers, so assume $\rho(\cdot\mid\omega)\ll\alpha$ for every $\omega$. Note also that $I(\rho)$ is always finite: since $q\in\Delta_{++}(\Omega)$, $\rho(f)=0$ forces $\rho(f\mid\omega)=0$ at every $\omega$, so that $f$ contributes nothing to either side under the convention $0\log0=0$. Then
$$
\sum_f \rho(f\mid\omega)\log\frac{\rho(f\mid\omega)}{\alpha(f)} = \sum_f \rho(f\mid\omega)\log\frac{\rho(f\mid\omega)}{\rho(f)} + \sum_f \rho(f\mid\omega)\log\frac{\rho(f)}{\alpha(f)}.
$$
Averaging both sides over $\omega$ with weights $q(\omega)$ and using $\sum_\omega q(\omega)\rho(f\mid\omega) = \rho(f)$ termwise on the right,
$$
\sum_\omega q(\omega)D_{\mathrm{KL}}\big(\rho(\cdot\mid\omega)\big\|\alpha\big) = I(\rho) + \sum_f \rho(f)\log\frac{\rho(f)}{\alpha(f)} = I(\rho) + D_{\mathrm{KL}}\big(\rho(\cdot)\big\|\alpha\big).
$$
Since $D_{\mathrm{KL}}(\rho(\cdot)\|\alpha)\ge0$ with equality iff $\alpha=\rho(\cdot)$, the left side is minimized over $\alpha$ exactly at $\alpha=\rho(\cdot)$, where it equals $I(\rho)$.
\end{proof}

\begin{proposition}[Genuine Rational Inattention in Closed Form]
\label{prop:ri-menu-bridge}
For every finite menu $F$, every $q\in\Delta_{++}(\Omega)$, and every $\theta>0$,
\begin{equation}
    \mathcal{RI}(q;F,\theta) = \theta \max_{\alpha\in\Delta(F)} \mathbb{E}_q\Big[ \log \mathbb{E}_\alpha\big[\exp(u(f(\omega))/\theta)\big] \Big],
    \label{eq:ri-menu-closed-form}
\end{equation}
where, for each $\omega$, $\mathbb{E}_\alpha[\,\cdot\,]$ denotes $\sum_{f\in F}\alpha(f)\exp(u(f(\omega))/\theta)$. Any maximizing pair coincides at the optimum: if $\alpha^*$ attains the maximum on the right and $\rho^*$ attains the maximum in \eqref{eq:ri-menu}, then $\alpha^*=\rho^*(\cdot)$ and $\rho^*(f\mid\omega)\propto\alpha^*(f)\exp(u(f(\omega))/\theta)$. Neither maximizer need be unique --- duplicate acts in $F$ already break uniqueness.
\end{proposition}
\begin{proof}
By Lemma \ref{lem:mi-variational}, $-\theta I(\rho) = \max_{\alpha\in\Delta(F)}\big\{-\theta\sum_\omega q(\omega) D_{\mathrm{KL}}(\rho(\cdot\mid\omega)\|\alpha)\big\}$ for every $\rho$, so \eqref{eq:ri-menu} becomes a joint maximization,
$$
\mathcal{RI}(q;F,\theta) = \max_{\rho}\max_{\alpha\in\Delta(F)} \sum_\omega q(\omega) \Big\{ \sum_f \rho(f\mid\omega)\,u(f(\omega)) - \theta\, D_{\mathrm{KL}}\big(\rho(\cdot\mid\omega)\big\|\alpha\big) \Big\},
$$
with no minimization anywhere, so the order of the two maxima is immaterial and, for fixed $\alpha$, the inner maximization over $\rho$ separates across $\omega\in\Omega$ (each $\rho(\cdot\mid\omega)\in\Delta(F)$ is chosen independently, and $q(\omega)\ge0$ are fixed weights). Fix $\alpha$ and $\omega$. Lemma \ref{lem:dv} requires a full-support reference measure, and $\alpha$ may lie on the boundary of $\Delta(F)$ --- as the equality discussion below shows it sometimes must. This costs nothing: writing $F_\alpha:=\mathrm{supp}\,\alpha$, we have $D_{\mathrm{KL}}(p\|\alpha)=+\infty$ unless $p\ll\alpha$, so the maximum over $\Delta(F)$ equals the maximum over $\Delta(F_\alpha)$, on which $\alpha$ has full support. Applying Lemma \ref{lem:dv} there, with $\Omega$ relabeled as $F_\alpha$, $q$ relabeled as $\alpha$, and $u(f(\omega))$ (a fixed real number for each $f$, since $\omega$ is held fixed) in place of $u(f)$, together with $u\mapsto-u$ to convert the min in \eqref{eq:anchor} to a max,
$$
\max_{p\in\Delta(F)} \left\{ \mathbb{E}_p\big[u(f(\omega))\big] - \theta D_{\mathrm{KL}}(p\|\alpha) \right\} = \theta \log \mathbb{E}_\alpha\big[\exp(u(f(\omega))/\theta)\big],
$$
attained uniquely at $p^*(f) \propto \alpha(f)\,\exp(u(f(\omega))/\theta)$. Taking $\rho(\cdot\mid\omega):=p^*$ at every $\omega$ and summing the resulting identity against $q(\omega)$ gives
\begin{align*}
    \mathcal{RI}(q;F,\theta)
    &= \max_{\alpha\in\Delta(F)} \sum_\omega q(\omega)\,\theta\log\mathbb{E}_\alpha\big[\exp(u(f(\omega))/\theta)\big] \\
    &= \theta\max_{\alpha\in\Delta(F)}\mathbb{E}_q\Big[\log\mathbb{E}_\alpha\big[\exp(u(f(\omega))/\theta)\big]\Big],
\end{align*}
which is \eqref{eq:ri-menu-closed-form}. For the final claim, write $J(\rho,\alpha):=\sum_\omega q(\omega)\{\sum_f\rho(f\mid\omega)u(f(\omega))-\theta D_{\mathrm{KL}}(\rho(\cdot\mid\omega)\|\alpha)\}$, so that any maximizing pair $(\rho^*,\alpha^*)$ attains $\max_{\rho,\alpha}J$. In particular $\alpha^*$ maximizes $\alpha\mapsto J(\rho^*,\alpha)$ with $\rho^*$ held fixed. But that map equals a constant depending only on $\rho^*$ minus $\theta\sum_\omega q(\omega)D_{\mathrm{KL}}(\rho^*(\cdot\mid\omega)\|\alpha)$, which by Lemma \ref{lem:mi-variational} is uniquely maximized at $\alpha=\rho^*(\cdot)$. Hence $\alpha^*=\rho^*(\cdot)$.
\end{proof}

Equation \eqref{eq:ri-menu-closed-form} is exactly the ``expected-log'' value that appears as \citeauthor{RSS_2023}'s (\citeyear{RSS_2023}) equation (7); reading the menu as a set of assets, their Proposition 2 identifies the solutions of the rational-inattention problem with the growth-optimal (Kelly) portfolios, by the same Donsker--Varadhan mechanics. They note that this equivalence is in turn mathematically equivalent to Lemma 2 of \citet{MM_2015}. The derivation above reaches \eqref{eq:ri-menu-closed-form} directly from Lemma \ref{lem:dv}, already established for the main text, rather than importing the portfolio analogy as a black box. Note that \citeauthor{RSS_2023} normalize the entropy coefficient to one and write $p$ for the reference prior and $q$ for the stochastic choice rule, the reverse of the convention used here. Its defining feature --- an outer expectation over $\omega$ of an inner log of an $\alpha$-mixture over $f\in F$ --- is structurally different from $V$ and $W_{\mathrm{RI}}$, both of which are a single log-sum-exp (equivalently, KL-tilt) over $\Omega$ with no menu-mixture inside and no further expectation over $\omega$ outside; this is the sense in which genuine Rational Inattention is not a relabeling of \eqref{eq:wri}, and why an exact identity between them should not be expected once $|F|>1$.

\subsection{A Sandwich Bound}

Although $\mathcal{RI}(q;F,\theta)$ and $\max_{f\in F}W_{\mathrm{RI}}(f)$ are generically distinct once $|F|>1$, they are not unrelated: the same menu $F$ bounds one below and above the other, with $V$ closing off the other side.

\begin{proposition}[$V$ and $W_{\mathrm{RI}}$ Sandwich Genuine Rational Inattention]
\label{prop:ri-sandwich}
For every finite menu $F$, every $q\in\Delta_{++}(\Omega)$, and every $\theta>0$,
\begin{equation}
    \max_{f\in F} V(f) \;\le\; \max_{f\in F}\mathbb{E}_q[u(f)] \;\le\; \mathcal{RI}(q;F,\theta) \;\le\; \max_{f\in F} W_{\mathrm{RI}}(f).
    \label{eq:sandwich}
\end{equation}
All four quantities coincide when $|F|=1$ (Proposition \ref{prop:ri-gap}). For $|F|>1$ the second and third inequalities may be strict or may bind; the proof records exactly when each holds with equality.
\end{proposition}
\begin{proof}
\emph{First inequality.} By Lemma \ref{lem:dv}, $V(f)\le\mathbb{E}_q[u(f)]$ for every act $f$ (Proposition \ref{prop:ri-gap}'s proof, applied to each $f\in F$ individually), so taking the maximum over $f\in F$ on both sides preserves the inequality.

\emph{Second inequality.} Let $f^*\in\arg\max_{f\in F}\mathbb{E}_q[u(f)]$ and take $\rho^{\deg}(f\mid\omega):=\mathbf{1}\{f=f^*\}$ for every $\omega$ (choose $f^*$ regardless of the state). Then $I(\rho^{\deg})=0$, since $\rho^{\deg}(\cdot\mid\omega)=\rho^{\deg}(\cdot)=\delta_{f^*}$ for every $\omega$, so $\rho^{\deg}$ is feasible in \eqref{eq:ri-menu} at zero cost and attains $\sum_\omega q(\omega)u(f^*(\omega)) = \mathbb{E}_q[u(f^*)] = \max_{f\in F}\mathbb{E}_q[u(f)]$; since $\mathcal{RI}(q;F,\theta)$ maximizes over \emph{all} feasible $\rho$, it is at least this value.

\emph{Third inequality.} By \eqref{eq:ri-menu-closed-form} and Jensen's inequality (twice), for any $\alpha\in\Delta(F)$,
\begin{align*}
    \mathbb{E}_q\Big[\log\mathbb{E}_\alpha\big[\exp(u(f(\omega))/\theta)\big]\Big]
    &\;\le\; \log\,\mathbb{E}_q\Big[\mathbb{E}_\alpha\big[\exp(u(f(\omega))/\theta)\big]\Big] \\
    &\;=\; \log \sum_{f\in F}\alpha(f)\,\mathbb{E}_q\big[\exp(u(f)/\theta)\big] \\
    &\;\le\; \log \max_{f\in F}\mathbb{E}_q\big[\exp(u(f)/\theta)\big],
\end{align*}
where the first step is Jensen's inequality applied to the concave function $\log(\cdot)$ (over the randomness in $\omega\sim q$, for fixed $\alpha$), the middle step is Fubini (both expectations are finite sums), and the last step bounds a convex combination over $f\in F$ by its largest term. The right side does not depend on $\alpha$, so taking $\max_\alpha$ on the left preserves the bound: $\max_\alpha\mathbb{E}_q[\log\mathbb{E}_\alpha[\exp(u(f(\omega))/\theta)]] \le \log\max_{f\in F}\mathbb{E}_q[\exp(u(f)/\theta)] = \max_{f\in F}\log\mathbb{E}_q[\exp(u(f)/\theta)]$. Multiplying by $\theta$ and using \eqref{eq:ri-menu-closed-form} on the left and, by Lemma \ref{lem:dv} applied to $-u(f)$ (as in the proof of Proposition \ref{prop:ri-gap}), $W_{\mathrm{RI}}(f) = \theta\log\mathbb{E}_q[\exp(u(f)/\theta)]$ on the right, gives $\mathcal{RI}(q;F,\theta) \le \max_{f\in F}W_{\mathrm{RI}}(f)$.

\emph{Collapse at $|F|=1$, and the equality cases.} When $F=\{f\}$, $\mathcal{RI}(q;F,\theta)=\mathbb{E}_q[u(f)]$ (noted after \eqref{eq:ri-menu}), so all four terms in \eqref{eq:sandwich} coincide with $\mathbb{E}_q[u(f)]$, consistently with Proposition \ref{prop:ri-gap}. For $|F|>1$ both middle inequalities can still bind, and we record exactly when.

The second is an equality precisely when the rational-inattention problem is solved by a state-independent rule, that is, when $\rho^{\deg}$ itself attains the maximum in \eqref{eq:ri-menu}. It is therefore not enough that $\argmax_{f\in F} u(f(\omega))$ vary with $\omega$: the attention cost must also be low enough relative to the payoff spread for a state-contingent rule to pay for itself. For $\Omega=\{\omega_1,\omega_2\}$ with $q=(0.7,0.3)$, $u(f_1)=(1,0)$ and $u(f_2)=(0,1)$, the optimum acquires no information at all once $\theta$ exceeds approximately $1.18$, and the second inequality is then an equality even though the state-optimal act varies with $\omega$.

The third is an equality precisely when both of the last two steps of its proof bind --- the application of Jensen's inequality and the bounding of a convex combination by its largest term --- that is, when
\begin{align*}
    \mathrm{supp}\,\alpha^* \ &\subseteq\ \argmax_{f\in F}\ \mathbb{E}_q\big[\exp(u(f)/\theta)\big], \\
    \omega\ \mapsto &\sum_{f\in F}\alpha^*(f)\exp\big(u(f(\omega))/\theta\big)\ \text{ is $q$-almost surely constant.}
\end{align*} (These conditions are stated at some, equivalently any, maximizer $\alpha^*$.) Symmetric menus satisfy both simultaneously: for $\Omega=\{\omega_1,\omega_2\}$ with $q$ uniform, $u(f_1)=(1,0)$ and $u(f_2)=(0,1)$, the maximizing $\alpha^*$ is uniform, the displayed mixture does not vary with $\omega$, and $\mathcal{RI}(q;F,\theta)=\max_{f\in F}W_{\mathrm{RI}}(f)$ at every $\theta>0$. We do not characterize the strict cases beyond this; \eqref{eq:sandwich} is what Proposition \ref{prop:ri-sandwich} asserts, and the two displayed conditions are exactly when it holds with equality.
\end{proof}

Proposition \ref{prop:ri-sandwich} is the precise sense in which $V$ and $W_{\mathrm{RI}}$, though neither is itself a genuine Rational Inattention value even once a menu is introduced, are not unrelated to one: maximized over the same menu, they are respectively a lower and an upper bound on it, meeting only in the degenerate single-act case of Proposition \ref{prop:ri-gap}. 
The gap in \eqref{eq:sandwich} captures the value of \emph{state-contingent} attention (absent for singletons, per Proposition \ref{prop:ri-gap}). Closing it requires abandoning the uniquely compatible log-sum-exp form (Theorem \ref{thm:collapse}(c) and Proposition \ref{prop:kl-uniqueness}).
We leave a full behavioral characterization of $\mathcal{RI}(q;F,\theta)$ itself, in the spirit of the entropic structure used above, for future work.

%


\section*{Acknowledgements}
First Draft: August 28, 2026. All remaining errors are my own.

\section*{Declaration of competing interest}
The author declares that he has no known competing financial interests or personal relationships that could have appeared to influence the work reported in this paper.

\section*{Funding}
This research received no specific grant from funding agencies in the public, commercial, or not-for-profit sectors.

\section*{Declaration of generative AI and AI-assisted technologies in the manuscript preparation process}
During the preparation of this work the author used Claude (Anthropic) and Google Gemini for editing and revising the text, for the Python code used in the numerical examples reported in Section \ref{sec:results}, and for checking proofs and references. All results, proofs, and references were verified by the author, who reviewed and edited the content as needed and takes full responsibility for the content of the published article.

\bibliographystyle{econ-econometrica.bst}
\bibliography{literature}

\end{document}